\documentclass[11pt]{article}
\usepackage[utf8]{inputenc}
\usepackage[margin=1in]{geometry}

\usepackage{sty}
\usepackage{tikz}
\usetikzlibrary{calc}

\title{Spectral Methods from Tensor Networks}

\makeatletter
\renewcommand*{\@fnsymbol}[1]{\ensuremath{\ifcase#1\or *\or \ddagger\or
    \mathsection\or \mathparagraph\or \|\or **\or \dagger\dagger
    \or \ddagger\ddagger \else\@ctrerr\fi}}
\makeatother

\usepackage{authblk}
\author[1,2]{Ankur Moitra%
\thanks{
 Email: {\tt moitra@mit.edu}. This work was supported in part by NSF CAREER Award CCF-1453261, NSF Large CCF-1565235, a David and Lucile Packard Fellowship, an Alfred P. Sloan Fellowship and an ONR Young Investigator Award.}%
}
\author[3]{Alexander S.\ Wein%
\thanks{Email: {\tt awein@cims.nyu.edu}. This work was supported in part by NSF grant DMS-1712730 and by the Simons Collaboration on Algorithms and Geometry.}%
}
\affil[1]{Department of Mathematics, Massachusetts Institute of Technology}
\affil[2]{Computer Science and Artificial Intelligence Lab, Massachusetts Institute of Technology}
\affil[3]{Department of Mathematics, Courant Institute of Mathematical Sciences, New York University}

\begin{document}

\maketitle

\begin{abstract}
A tensor network is a diagram that specifies a way to ``multiply'' a collection of tensors together to produce another tensor (or matrix). Many existing algorithms for tensor problems (such as tensor decomposition and tensor PCA), although they are not presented this way, can be viewed as spectral methods on matrices built from simple tensor networks. In this work we leverage the full power of this abstraction to design new algorithms for certain continuous tensor decomposition problems. 

An important and challenging family of tensor problems comes from orbit recovery, a class of inference problems involving group actions (inspired by applications such as cryo-electron microscopy). Orbit recovery problems over finite groups can often be solved via standard tensor methods. However, for infinite groups, no general algorithms are known. We give a new spectral algorithm based on tensor networks for one such problem: continuous multi-reference alignment over the infinite group $\mathrm{SO}(2)$. Our algorithm extends to the more general heterogeneous case.
\end{abstract}

\newpage

\section{Introduction}

% success of tensor methods
Algorithms for decomposing low-rank tensors have had a wide range of applications in machine learning and statistics. They can be leveraged to give efficient algorithms for phylogenetic reconstruction \cite{mosselroch}, topic modeling \cite{anandkumartopic}, community detection \cite{anandkumarcommunity}, independent component analysis \cite{moitra-notes} and learning various mixture models \cite{hsukakade, jainoh}. However there are important families of problems where the low-order moment tensors are known to achieve statistically-optimal rates of estimation but there are no known \emph{efficient} algorithms for finding the parameters from the moments. 

% tensor decomposition
The familiar symmetric third-order tensor decomposition problem asks: Given a $p \times p \times p$ low-rank tensor of the form
$$T = \sum_{i =1}^r a_i^{\otimes 3}$$
can we recover the vectors $a_1,\ldots,a_r \in \mathbb{R}^p$? When $r \leq p$ it is called the {\em undercomplete} case and when $r > p$ it is called the {\em overcomplete} case. In the undercomplete case, Jennrich's algorithm (see \cite{moitra-notes}) gives a polynomial time algorithm based on generalized eigendecompositions that works provided that the vectors $a_1,\ldots,a_r$ are linearly independent. In the overcomplete case, a line of work has culminated in a polynomial time algorithm that works when the vectors $a_i$ are random (i.i.d.\ Gaussian) and $r \lesssim p^{3/2}$ \cite{GM-tensor-sos,HSSS-fast-sos,MSS-tensor-sos}. In applications, the vectors $a_1,\ldots,a_r$ represent the parameters of a model we would like to learn and $T$ represents moments of the distribution specified by the model whose entries we can estimate from samples.

% continuous tensor decomp
However, in some applications the parameters are not uniquely defined, except up to equivalence under some {\em continuous group action}. This leads to a new sort of problem that we call {\em orbit tensor decomposition} in which we want to recover a vector $\theta \in \mathbb{R}^p$ given a tensor of the form
$$T = \int_{A \in \mathcal{A}} (A\theta)^{\otimes 3} dA$$
where $\mathcal{A}$ is a known, possibly infinite, set of $p \times p$ matrices (equipped with a measure over which to integrate). We assume furthermore that $\mathcal{A}$ possesses a particular group symmetry which causes nonuniqueness of the solution: $\theta$ and $A\theta$ are equally-good solutions for any $A \in \mathcal{A}$. There are important real-world applications such as cryo-electron microscopy (cryo-EM) \cite{AdrDubLep84,ss-cryo,nog-cryo} and multi-reference alignment (MRA) \cite{ZwaHeiGel03,PitZurAmo05,mra-sdp,AbbPerSin17,BanRigWee17,PerWeeBan17,mra-bispectrum,spectral-mra} where these sort of tensor decomposition problems arise when using the method of moments. Here $A$ is a random rotation of a two- or three-dimensional signal whose orientation we cannot control when we are measuring it. Despite considerable interest in such problems there are few algorithms with provable guarantees, in large part because working with the symmetries of the group is challenging {\em algorithmically}.

% continuous MRA
We will focus on the \emph{continuous multi-reference alignment} (continuous MRA) problem which can be described as follows. The goal is to recover a signal $\theta$ which is a real-valued function on the unit circle in $\mathbb{R}^2$. We assume $\theta$ is band-limited so that in the Fourier basis we can think of $\theta$ as a finite-dimensional vector $\theta \in \mathbb{R}^p$. The compact group $G = \mathrm{SO}(2)$ (rotations in the plane) acts on $\theta$ by rotating the signal around the unit circle. For $g \in G$ and $\theta \in \mathbb{R}^p$ we denote the result of the rotation as $g \cdot \theta \in \mathbb{R}^p$. Now we observe many independent samples of the form $y_i = g_i \cdot \theta + \xi_i$ where $g_i$ is a uniformly random element of $\mathrm{SO}(2)$ and $\xi_i$ is i.i.d.\ Gaussian noise. {\em In other words, we observe many copies of the true signal that are both noisy and randomly-rotated.} It is known that for this problem (and a large class of similar problems), optimal sample complexity in the large-noise limit is achieved by the method of moments \cite{AbbPerSin17,BanRigWee17,orbit-recovery,group-channel}. First we use the samples to estimate the third moment $T = \int_{g \in G} (g \cdot \theta)^{\otimes 3} dg$. Recovering $\theta$ (up to equivalence under group action) is now an instance of the orbit tensor decomposition problem from above.

% symmetry; failure of Jennrich
Existing tensor methods fail because (i) $T$ is no longer low-rank. In fact $T$ has an infinite number of components and when $\theta$ is generic would plausibly have essentially full rank. We can no longer hope to decompose $T$ by finding a rank-one term that we can subtract off and lower the rank. Instead, we need to find a continuous collection of rank-one tensors at once! (ii) We can only hope to recover the {\em orbit} of  $\theta$, i.e.\ to recover a vector that (approximately) lies in the orbit $\{g \cdot \theta\;:\;g \in G\}$. This symmetry implies that any finite-rank decomposition of the tensor cannot be unique, which seems to rule out many spectral methods such as Jennrich's algorithm (whose analysis relies on having a unique decomposition).

We remark that for \emph{discrete multi-reference alignment} (discrete MRA) where $G$ is a {\em finite} group of rotations of order $p$, these issues do not arise. In fact, the samples $y_i$ can be thought of as coming from a mixture of $p$ spherical Gaussians where the centers are related (in that they are rotations of each other). By ignoring these interrelationships and learning the distribution as a mixture of spherical Gaussians via tensor decomposition, it is possible to obtain algorithms with provable guarantees \cite{PerWeeBan17}. In contrast, continuous MRA is a continuous mixture model where we crucially must exploit the relationship between the (infinitely-many) centers. The continuous nature of our problem poses a fundamental challenge for applying tensor methods. To overcome this, we will first randomly break the symmetry and then apply a spectral method that resembles a tailor-made variant of the tensor power method.

% our contributions
In this paper, we leverage this methodology to give a polynomial-time algorithm for \emph{list recovery} for the continuous MRA problem and for its so-called \emph{heterogeneous} generalization in which there are multiple true signals $\theta^1,\ldots,\theta^K \in \mathbb{R}^p$ and each sample comes from a random one of them. Here \emph{list recovery} means that we output a list of polynomially-many candidate vectors such that every true signal is well correlated with at least one candidate. To achieve this, we need to delicately exploit symmetries in the orbit of each $\theta^k$, but cope with the fact that the orbits of different components are unrelated. More broadly, our success gives us hope that our methodology for designing tensor spectral methods can be adapted to a wide variety of problems that have thus far resisted attack. As in work on overcomplete tensor decomposition \cite{GM-tensor-sos,HSSS-fast-sos,MSS-tensor-sos}, our analysis assumes that the signals $\theta^k$ are drawn at random (i.i.d.\ Gaussian). To the best of our knowledge, our algorithm provides the first polynomial-time solution to an orbit recovery problem over an infinite group, other than a few special cases that admit \emph{ad hoc} closed-form solutions (see Section~\ref{sec:freq-march}). In particular, we give the first polynomial-time solution to a \emph{heterogeneous} orbit recovery problem over an infinite group.

We now motivate and describe our approach for the continuous MRA problem. Many existing methods for overcomplete tensor decomposition are based on the idea of finding a vector $v \in \mathbb{R}^p$ that maximizes the cubic form $\langle T,v^{\otimes 3} \rangle = \langle \sum_{i=1}^r a_i^{\otimes 3},v \rangle$. If the $a_i$ are random, it can be shown that approximately, the maximizers of $\langle T,v^{\otimes 3} \rangle$ are $a_1,\ldots,a_r$ provided $r \lesssim p^{3/2}$ \cite{GM-tensor-sos}. A popular heuristic for optimizing $\langle T,v^{\otimes 3} \rangle$ over unit vectors is the tensor power method, in which we iteratively update $v \in \mathbb{R}^p$ according to
\begin{equation}\label{eq:power-method}
v_i \leftarrow \sum_{jk} T_{ijk} v_j v_k.
\end{equation}
Similarly to the matrix power method, the intuition here is that by ``multiplying'' the tensor by itself, we are repeatedly amplifying the signal without having the noise build up too much. There are rigorous guarantees for this non-convex method for random overcomplete tensor decomposition, but require a very warm start \cite{tensor-power-1,tensor-power-2}.

Perhaps fortuitously, unlike the matrix case there are many different ways that one can ``multiply'' third-order tensors together to create other ``power methods.'' A \emph{tensor network} is a diagram that specifies a recipe for multiplying a collection of tensors together. This concept has been used in areas such as quantum physics \cite{tensor-net}. Tensor network notation is illustrated in Figure~\ref{fig:net-intro} and will be central to our work. One of our key observations is that, although they were not explained this way, many existing tensor methods in the literature can be re-interpreted as spectral methods on matrices derived from tensor networks. In particular, the spectral method of \cite{HSSS-fast-sos} for random overcomplete tensor decomposition is based on the tensor network shown in Figure~\ref{fig:net-intro}(c); this method is a starting point for our work. In Appendix~\ref{app:tensor-pca} we catalog related results for the \emph{tensor PCA} problem and how they can also be described as coming from certain tensor networks (which are depicted in Figure~\ref{fig:tensor-pca}).

\begin{figure}[!ht]\centering
\begin{minipage}{0.8\textwidth}\centering
\includegraphics[width=0.8\textwidth]{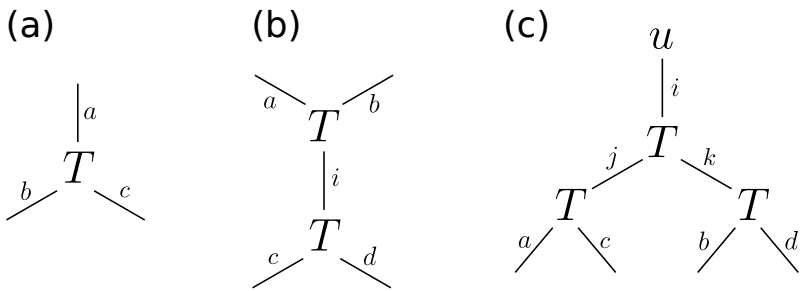}
\caption{An introduction to tensor network notation. {\bf (a)} A single copy of the third-order tensor $T$ (with entries $T_{abc}$) has three \emph{legs}, one for each mode. {\bf (b)} Two copies of $T$ connected by contracting (summing over) the index $i$. The result is the fourth-order tensor $B_{abcd} = \sum_i T_{abi}T_{cdi}$. {\bf (c)} The spectral method in \cite{HSSS-fast-sos} uses the $(\{a,b\},\{c,d\})$-flattening of this tensor network (which is a $p^2 \times p^2$ matrix). We explain this in more detail in Section~\ref{sec:techniques}. Here $u$ is a random vector.}
\label{fig:net-intro}
\end{minipage}
\end{figure}

The tensor network abstraction gives us freedom to explore more complicated tensor networks, which helps us cope with the symmetries of continuous MRA. Ultimately we will use the tensor network in Figure~\ref{fig:net-9}. We will show that with decent probability over a random tensor $u$, the top eigenvector of the associated matrix is close to a vector in the orbit of $\theta$. To accomplish this, we will employ the trace moment method which, in our setting, gives us a way to spectrally bound a certain noise term by counting certain \emph{valid} labelings of the edges of a much larger tensor network that is obtained by stringing together many copies of Figure~\ref{fig:net-9}. The constraints imposed on a valid labeling are dictated by the $\mathrm{SO}(2)$ group structure.

We remark that our tensor $T$ is quite sparse in the Fourier domain. (This is in stark contrast to the situation in random overcomplete tensor decomposition or tensor PCA.) In particular, $T$ is $p \times p \times p$ but only supported on the $\sim p^2$ entries $T_{ijk}$ for which $i+j+k=0$. This comes from the fact that due to integrating over the group action, $T$ is a projection of $\theta^{\otimes 3}$ onto a particular subspace (namely the span of the degree-three invariant polynomials; see \cite{orbit-recovery}). The above sparsity pattern influences the combinatorics of the trace moment method. In particular, our valid labelings (discussed above) require that the three incoming legs to each copy of the tensor sum to zero. This is a rather different sort of combinatorics problem than typically arises in applications of the trace moment method to random matrix theory, and at a high-level, is why we need such a complex tensor network. In Appendix~\ref{app:design}, we discuss in more detail the considerations behind choosing the particular tensor network in Figure~\ref{fig:net-9}.

% fig:net-9
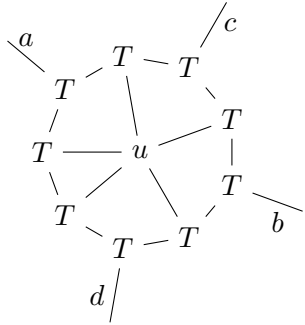
\begin{figure}\centering
\begin{minipage}{0.8\textwidth}\centering
\begin{tikzpicture}
\def\ra{1.3}  % inner radius
\def\rad{2.3}  % outer radius
\node (u) at (0,0) {$u$};
\node (T1) at ({-\ra*cos(1*360/9)},{\ra*sin(1*360/9)}) {$T$};
\node (T2) at ({-\ra*cos(2*360/9)},{\ra*sin(2*360/9)}) {$T$};
\node (T3) at ({-\ra*cos(3*360/9)},{\ra*sin(3*360/9)}) {$T$};
\node (T4) at ({-\ra*cos(4*360/9)},{\ra*sin(4*360/9)}) {$T$};
\node (T5) at ({-\ra*cos(5*360/9)},{\ra*sin(5*360/9)}) {$T$};
\node (T6) at ({-\ra*cos(6*360/9)},{\ra*sin(6*360/9)}) {$T$};
\node (T7) at ({-\ra*cos(7*360/9)},{\ra*sin(7*360/9)}) {$T$};
\node (T8) at ({-\ra*cos(8*360/9)},{\ra*sin(8*360/9)}) {$T$};
\node (T9) at ({-\ra*cos(9*360/9)},{\ra*sin(9*360/9)}) {$T$};
\coordinate (a) at ({-\rad*cos(1*360/9)},{\rad*sin(1*360/9)});
\coordinate (c) at ({-\rad*cos(3*360/9)},{\rad*sin(3*360/9)});
\coordinate (b) at ({-\rad*cos(5*360/9)},{\rad*sin(5*360/9)});
\coordinate (d) at ({-\rad*cos(7*360/9)},{\rad*sin(7*360/9)});
\draw (T1) -- (T2);
\draw (T2) -- (T3);
\draw (T3) -- (T4);
\draw (T4) -- (T5);
\draw (T5) -- (T6);
\draw (T6) -- (T7);
\draw (T7) -- (T8);
\draw (T8) -- (T9);
\draw (T9) -- (T1);
\draw (T1) -- (a) node [midway,above] {$a$};
\draw (T2) -- (u);
\draw (T3) -- (c) node [midway,right] {$c$};
\draw (T4) -- (u);
\draw (T5) -- (b) node [midway,below] {$b$};
\draw (T6) -- (u);
\draw (T7) -- (d) node [midway,left] {$d$};
\draw (T8) -- (u);
\draw (T9) -- (u);
\end{tikzpicture}
\caption{In this paper we will analyze a spectral method on the $p^2 \times p^2$ matrix given by the $(\{a,b\},\{c,d\})$-flattening of the tensor network shown here. Here $u$ is a random order-5 tensor.}
\label{fig:net-9}
\end{minipage}
\end{figure}

\section{Orbit recovery problems}

\subsection{Problem statement}
\label{sec:prob-statement}

We now formally define \emph{orbit recovery} problems including continuous MRA. These are a class of problems for which the method of moments gives rise to an \emph{orbit tensor decomposition} problem.

Let $G$ be a \emph{compact group}. We do not formally define the notion of a compact group here, but some examples of interest include: (i) any finite group, such as the symmetric group $S_L$ (permutations of $\{1,\ldots,L\}$) and the cyclic group $\mathbb{Z}/L$, (ii) 2-dimensional rotations $\mathrm{SO}(2)$, and (iii) 3-dimensional rotations $\mathrm{SO}(3)$.

Let $G$ \emph{act linearly} on $\mathbb{R}^p$. A linear action means that each group element $g \in G$ has an associated matrix $\rho(g) \in \mathbb{R}^{p \times p}$ by which it acts on $\mathbb{R}^p$ (via matrix multiplication): for $\theta \in \mathbb{R}^p$ we write $g \cdot \theta = \rho(g) \theta$. The matrices must be consistent with the group structure, i.e.\ $\rho(gh) = \rho(g)\rho(h)$ and $\rho(e) = I$ where $e \in G$ is the identity.

Given a compact group $G$ acting linearly on $\mathbb{R}^p$, we define the associated \emph{orbit recovery} problem \cite{orbit-recovery} as follows. (This has also been called the \emph{group action channel} \cite{group-channel}.) For $i = 1,\ldots,n$ we observe
$$y_i = g_i \cdot \theta + \xi_i$$
where $\theta \in \mathbb{R}^p$ is the unknown signal, $g_i$ is drawn from \emph{Haar measure} (the ``uniform distribution'') on $G$, and $\xi_i \sim \mathcal{N}(0,\sigma^2 I)$. The random variables $g_i, \xi_i$ are all independent. The goal is to estimate $\theta$ up to group action, i.e.\ to output an estimator close to the orbit $\{g \cdot \theta \;:\; g \in G\}$ of $\theta$.

The following are some motivating examples of orbit recovery problems.
\begin{itemize}
    \item {\bf (Discrete) multi-reference alignment (MRA)} \cite{mra-sdp,AbbPerSin17,BanRigWee17,PerWeeBan17,mra-bispectrum,spectral-mra}: This is the case where $G$ is the cyclic group $\mathbb{Z}/p$ acting on $\mathbb{R}^p$ via cyclic permutation. Formally, for $g \in \mathbb{Z}/p$ (integers mod $p$), let $(g \cdot \theta)_i = \theta_{i-g\text{ (mod }p\text{)}}$. This captures the problem where we see many noisy copies of the same discrete signal, each with a different offset. This has applications in signal processing \cite{ZwaHeiGel03,PitZurAmo05} and structural biology \cite{Dia92,TheSte12}. We refer to the above problem as \emph{discrete MRA} in contrast to \emph{continuous MRA} which will be defined later.
    \item {\bf Cryo-electron microscopy (cryo-EM)} \cite{AdrDubLep84,ss-cryo,nog-cryo,orbit-recovery}: Cryo-EM is a popular biological imagining technique used to deduce the 3-dimensional structure of a large molecule such as a protein. This method was awarded the 2017 Nobel Prize in Chemistry. The method produces data in the form of many noisy 2-dimensional images of the 3-dimensional molecule, but in each image the molecule is rotated to an unknown orientation in 3-dimensional space. Here we think of $\theta \in \mathbb{R}^p$ as a representation of the molecule in some fixed basis (see \cite{orbit-recovery} for a precise definition). The group is $G = \mathrm{SO}(3)$ acting by rotating the molecule. This is a generalization of the orbit recovery problem where we observe $y_i = \Pi(g_i \cdot \theta) + \xi_i$ where $\Pi$ is a fixed linear operator, namely the mapping from a 3-dimensional molecule to a 2-dimensional image.
\end{itemize}

We will consider the \emph{heterogeneous} extension of orbit recovery. This is motivated by cryo-EM in situations where there are multiple molecules (or multiple conformations of the same molecule) and each image contains an unknown one of them. Formally, there are $K$ true signals $\theta^1,\ldots,\theta^K \in \mathbb{R}^p$ and each sample takes the form
$$y_i = g_i \cdot \theta^{k_i} + \xi_i$$
where $k_i$ is drawn at random from $[K] = \{1,\ldots,K\}$. In general, one can consider an arbitrary distribution over $[K]$, but we will restrict ourselves to the case where $k_i$ is drawn uniformly from $[K]$. In the heterogeneous problem, the goal is to estimate $\theta^1,\ldots,\theta^K$ up to permutation and group action.

\subsection{Continuous MRA}\label{sec:continuousMRA}

In this paper we will focus on the (heterogeneous) \emph{continuous MRA} problem, as it is a simple example of an orbit recovery problem over an infinite group. Here we take the group to be $G = \mathrm{SO}(2)$, parametrized by angles $g \in [0,2\pi)$. (Haar measure is simply the uniform distribution on angles.) Let $p$ be even. The signal is $\theta \in \mathbb{R}^p$ with entries indexed by the ``frequencies'' $\pm j$ for $j \in [p/2] = \{1,2,\ldots,p/2\}$. We will denote this set of frequencies by $\pm [p/2] = \{-p/2,\ldots,-1,1,\ldots,p/2\}$ (note that $0$ is not included for convenience). The action of $G$ on $\mathbb{R}^p$ is block-diagonal with $2 \times 2$ blocks: $g \in G$ acts on $[\theta_j \; \theta_{-j}]^\top$ (with $j > 0$) via the matrix
$$\left(\begin{array}{cc}
\cos(j g) & -\sin(j g) \\
\sin(j g) & \cos(j g)
\end{array}\right).$$

\noindent It will sometimes be convenient to work in the Fourier basis: for $j > 0$,
\begin{equation}\label{eq:Delta}
\hat\theta_j = \frac{1}{\sqrt 2}(\theta_j + \im\, \theta_{-j}) \quad\text{and}\quad \hat\theta_{-j} = \frac{1}{\sqrt 2}(\theta_j - \im\, \theta_{-j})
\end{equation}
where $\im$ is the imaginary unit. If $\theta \sim \mathcal{N}(0,I/p)$, we have $\hat\theta_j \sim \mathcal{N}(0,1/(2p)) + \im\, \mathcal{N}(0,1/(2p))$ with $\hat\theta_{-j} = \overline{\hat\theta_j}$ (complex conjugate). In the Fourier basis, the action of $G$ is diagonal, with $g$ acting on $\hat \theta_j$ by the scalar $\exp(\im j g)$.

\subsection{Method of moments}

One method for approaching orbit recovery problems is to attempt to learn the unknown group elements $g_i$. This is the well-studied \emph{synchronization approach} \cite{singer-angular,ss-cryo,mra-sdp,nug,afonso-thesis,cryo-orthogonal,pwbm-amp}.

An alternative approach uses the \emph{method of moments}, which seeks to estimate $\theta$ directly from the moments of the samples without attempting to estimate the $g_i$. This was discovered first in the case of MRA \cite{AbbPerSin17,BanRigWee17,PerWeeBan17} and later extended to all groups \cite{orbit-recovery,group-channel}. This method is suited to the case where the noise $\sigma$ on each sample is very large but we get many samples; in this regime we cannot hope to accurately estimate $g_i$ but can still hope to recover $\theta$.

We now describe the method of moments more formally. Consider the heterogeneous problem with signals $\theta^1,\ldots,\theta^K \in \mathbb{R}^p$. In the method of moments we use the samples $y_i$ to estimate the moments
\begin{align*}
T_1(\{\theta^k\}) &= \Ex_{k,g}[g \cdot \theta^k] = \frac{1}{K} \sum_{k=1}^K \Ex_g[g \cdot \theta^k]\\
T_2(\{\theta^k\}) &= \Ex_{k,g}[(g \cdot \theta^k)(g \cdot \theta^k)^\top] = \frac{1}{K} \sum_{k=1}^K \Ex_g[(g \cdot \theta^k)(g \cdot \theta^k)^\top]\\
&\;\vdots\\
T_d(\{\theta^k\}) &= \Ex_{k,g}[(g \cdot \theta^k)^{\otimes d}] = \frac{1}{K} \sum_{k=1}^K \Ex_g[(g \cdot \theta^k)^{\otimes d}].
\end{align*}
Above, the expectation is over $k$ drawn uniformly from $[K]$ and $g$ drawn from Haar measure on $G$. It is possible to accurately estimate the moments $T_1,\ldots,T_d$ given roughly $n \sim \sigma^{2d}$ samples (recall $\sigma$ is the noise level) \cite{orbit-recovery,group-channel}. Thus we are interested in an inversion procedure that recovers $\{\theta^k\}$ (up to permutation and orbit) given $T_1,\ldots,T_d$, for $d$ as small as possible. General algebraic techniques exist for testing how large $d$ needs to be for this to be possible, but this does not necessarily give a polynomial-time algorithm to actually recover the signal from the moments \cite{orbit-recovery}. For many natural problems such as MRA and cryo-EM, it is known that $d = 3$ is sufficient (and necessary) \cite{AbbPerSin17,BanRigWee17,PerWeeBan17,orbit-recovery}.

It is known that the method of moments is statistically optimal in the limit $\sigma \to \infty$ (with the group, group action, and dimension $p$ fixed) in the following sense \cite{BanRigWee17,orbit-recovery,group-channel}. On one hand, $n \sim \sigma^{2d}$ samples are sufficient to estimate the moments $T_1,\ldots,T_d$. On the other hand, if two signals $\theta,\theta'$ (or more generally, two collections of $K$ heterogeneous signals) produce the same $T_1,\ldots,T_{d-1}$ then at least $n \sim \sigma^{2d}$ samples are statistically required in order to distinguish between $\theta$ and $\theta'$. In other words, if the method of moments requires moments up to $d$ then \emph{any} method requires at least $\sigma^{2d}$ samples.

For the case of continuous MRA, it is easiest to work with the moments in the Fourier domain: $\hat T_d(\{\theta^k\}) = \frac{1}{K} \sum_{k=1}^K \Ex_g[(g \cdot \hat\theta)^{\otimes d}]$ where the action of $g$ on $\theta$ is diagonal: identifying $g$ with an angle $g \in [0,2\pi)$ we have $(g \cdot \hat\theta)_j = \exp(\im j g) \hat\theta_j$ (where $\im$ is the imaginary unit). For $j_1,\ldots,j_d \in \pm [p/2]$ we can compute
\begin{align}
\hat T_d(\{\theta^k\})_{j_1,\ldots,j_d} &=
\frac{1}{K} \sum_{k=1}^K \Ex_g[(g \cdot \hat\theta^k)_{j_1} \cdots (g \cdot \hat\theta^k)_{j_d}] \nonumber \\
&= \frac{1}{K} \sum_{k=1}^K \Ex_g[\exp(\im j_1 g)\hat\theta^k_{j_1} \cdots \exp(\im j_d g)\hat\theta^k_{j_d}] \nonumber \\
&= \frac{1}{K} \sum_{k=1}^K \Ex_g[\exp(\im g (j_1 + \cdots + j_d))] \hat\theta^k_{j_1} \cdots \hat\theta^k_{j_d} \nonumber \\
&= \frac{1}{K} \sum_{k=1}^K \mathbbm{1}_{j_1 + \cdots + j_d=0} \,\hat\theta^k_{j_1} \cdots \hat\theta^k_{j_d}.
\label{eq:T}
\end{align}

\subsection{Efficient algorithms}

We have seen above that the optimal statistical procedure is to compute moments $T_i$ and to use these to solve for $\{\theta^k\}$ consistent with these moments. \emph{A priori}, this is a polynomial system of equations which cannot be solved efficiently. In this section we survey known polynomial-time methods for recovering the signal(s) from the moments in special cases.

\subsubsection{Frequency marching}
\label{sec:freq-march}

Both the discrete and continuous MRA problems admit a closed-form solution called \emph{frequency marching} in the \emph{homogeneous} case ($K=1$). These methods are limited in the sense that they rely heavily on the particular structure of MRA and do not seem to extend to other groups or to the heterogeneous case (even for $K=2$).

For discrete MRA, the frequency marching approach is described in \cite{mra-bispectrum}. An essentially-identical method works for continuous MRA, which we describe here.

Consider the homogeneous continuous MRA problem. The goal is to recover $\theta$ from $T_2(\theta)$ and $T_3(\theta)$ under the assumption that all Fourier coefficients of $\theta$ are nonzero. Recall the structure of moments (\ref{eq:T}). From $T_2$ we learn, for every $j \in [p/2]$, the value $\hat\theta_j \hat\theta_{-j} = \hat\theta_j \overline{\hat\theta_j} = |\hat\theta_j|$, i.e.\ we learn the magnitudes of the Fourier coefficients (the \emph{power spectrum}). It suffices to recover the phases. From $T_3$ we learn the value $\hat\theta_{j_1}\hat\theta_{j_2}\hat\theta_{j_3}$ for every $j_1,j_2,j_3 \in \pm[p/2]$ such that $j_1+j_2+j_3=0$ (the \emph{bispectrum}). Provided $\hat\theta_1 \ne 0$, each orbit has a unique representative such that the phase $\phi_1$ of $\hat\theta_1$ is 0. Thus we take $\phi_1 = 0$. Now use $\hat\theta_{-1} \hat\theta_{-1} \hat\theta_2$ to learn $\phi_2$, use $\hat\theta_{-1} \hat\theta_{-2} \hat\theta_3$ to learn $\phi_3$, and so on until we have learned all the phases.

Another problem that admits a similar closed-form solution (in the homogeneous case only) is cryo-ET (cryo-electron tomography), a variant of cryo-EM without the projection step \cite{orbit-recovery}. The cryo-EM problem remains open (even in the homogeneous case): there are no known polynomial-time algorithms with provable guarantees.

\subsubsection{Tensor decomposition}

Note that when $G$ is a finite group, the third moment $T_3$ takes the form
$$T_3(\{\theta^k\}) = \frac{1}{K |G|} \sum_{k=1}^K \sum_{g \in G} (g \cdot \theta^k)^{\otimes 3}$$
which is a low-rank tensor (of rank $K|G|$) and is thus amenable to standard tensor decomposition techniques. For homogeneous discrete MRA, $T_3$ is undercomplete and can be decomposed using Jennrich's algorithm \cite{PerWeeBan17}, thus recovering (all shifts of) the signal. For heterogeneous discrete MRA, $T_3$ is overcomplete (rank exceeds dimension) but Jennrich's algorithm can still be used if we are given a higher order moment tensor. For instance, if $K \le p/2$ then Jennrich's algorithm can be used to decompose $T_5$ \cite{PerWeeBan17}. However, estimating $T_5$ requires suboptimal sample complexity $n \sim \sigma^{10}$. If we assume $\theta^k$ are random (i.i.d.\ Gaussian) and $K \lesssim \sqrt{p}$, we can avoid this by using overcomplete methods to decompose $T_3$ \cite{alex-thesis}. This result is an adaptation of methods for random overcomplete tensor decomposition using the sum-of-squares hierarchy \cite{MSS-tensor-sos}. It is conjectured that $K \lesssim \sqrt{p}$ is optimal for efficient methods that use $T_3$ \cite{mra-het,alex-thesis}.

\begin{remark}
One property of (the analysis of) Jennrich's algorithm is that it is only guaranteed to work in cases where the tensor has a unique decomposition. This is a serious barrier to using Jennrich's algorithm for problems over infinite groups. If $G$ is infinite, we might still hope that $T_3 = \mathbb{E}_g [(g \cdot \theta)^{\otimes 3}]$ (or a higher-order moment) has a low-rank decomposition and that this decomposition tells us something about $\theta$. However, even if this were true, we could not use (the existing analysis of) Jennrich's algorithm to find such a decomposition because the decomposition would not be unique: if $T_3 = \sum_{i=1}^r a_i^{\otimes 3}$ then we also have $T_3 = \sum_{i=1}^r (g \cdot a_i)^{\otimes 3}$ for any $g \in G$. More generally, it seems that any spectral method (which attempts to recover the signal as an eigenvector of some matrix) cannot succeed unless it first breaks the symmetry; otherwise there are infinitely-many solutions but a matrix only has finitely-many eigenvectors. Our method will randomly break the symmetry and then use a spectral method.
\end{remark}

\section{Results and techniques}

\subsection{Notation}

We say an event occurs with \emph{high probability} if it has probability $1-o(1)$ (as $p \to \infty$). We say an event occurs with \emph{overwhelming probability} if it occurs with probability $1-1/\delta(p)$ where $\delta(p)$ grows faster than any polynomial in $p$ (i.e.\ for any $k \in \mathbb{N}$, $\delta(p) \ge \omega(p^k)$). The notation $\tilde O(\cdots)$ hides factors of $\log(p)$.

We write $[p] = \{1,2,\ldots,p\}$ and define $\pm [p/2]$ as in Section~\ref{sec:continuousMRA}. The $p \times p$ identity matrix is denoted $I_p$ or simply $I$. We use $\|\cdot \|$ to denote the spectral norm of a matrix. We use $\|\cdot \|_F$ to denote the Frobenius norm (of either a matrix or tensor). For a tensor $T$, we use e.g.\ $\|T\|_{\{a,b\},\{c,d\}}$ to denote the spectral norm of the $(\{a,b\},\{c,d\})$-flattening of $T$. The $(\{a,b\},\{c,d\})$-flattening of a 4-tensor $T \in (\mathbb{R}^p)^{\otimes 4}$ is the $p^2 \times p^2$ matrix $M_{ab,cd} = T_{abcd}$.

\subsection{Main result}

We now state our main result for continuous MRA.

\begin{theorem}[list recovery for heterogeneous continuous MRA]
\label{thm:mra-main}
Let $\theta^1,\ldots,\theta^K \in \mathbb{R}^p$ be drawn independently from $\mathcal{N}(0,I_p/p)$. Suppose we are given the tensor $\mathcal{T} = T + E \in (\mathbb{R}^p)^{\otimes 3}$ where $\|E\|_{\infty} \le K^{-8} p^{-4}/\mathrm{polylog}(p)$ and
$$T = \sum_{k=1}^K \,\Ex_g \left[(g \cdot \theta^k)^{\otimes 3}\right]$$
with $g$ drawn from Haar (uniform) measure on $\mathrm{SO}(2)$. For any $\varepsilon > 0$, there is an algorithm that runs in time $p^{O(1)/\varepsilon^4}$ and outputs a list of unit vectors $\tau_1,\ldots,\tau_L \in \mathbb{R}^p$ with $L = p^{O(1)/\varepsilon^4}$ that has the following guarantee. Suppose $K \le p^\delta$ for a universal constant $\delta > 0$. With high probability over both $\theta^1,\ldots,\theta^K$ and the algorithm's randomness, for every $k \in [K]$ there exists $i \in [L]$ such that $\langle \tau_i,\theta^k \rangle^2 \ge 1 - \varepsilon - o(1)$.
\end{theorem}

For any constant $\varepsilon$, our algorithm runs in polynomial time. To the best of our knowledge, this is the first polynomial-time algorithm for a heterogeneous orbit recovery problem over an infinite group. (A few homogeneous problems have frequency marching solutions; see Section~\ref{sec:freq-march}.) Moreover by Proposition~7.6 of \cite{orbit-recovery}, to compute $\mathcal{T}$ satisfying the above condition on $\|E\|_\infty$, it is sufficient to take $n = \tilde O(\sigma^6 K^{18} p^8)$ samples. This exhibits {\em statistically-optimal} dependence of $\sigma^6$ on the noise level. We do not attempt to optimize the constant $\delta$, but we expect that $K \sim \sqrt{p}$ is optimal; see Appendix~\ref{app:design}. 

Our algorithm produces a list of candidate solutions but we do not analyze how to hypothesis test to select the correct solution(s) from the list. We leave this as an open question for future work. Heuristically, in the homogeneous case, one can evaluate a candidate solution $\tau$ by comparing $T_2(\tau)$ and $T_3(\tau)$ to our estimates for the true moments $T_2(\theta), T_3(\theta)$. In the heterogeneous case, we want to find vectors $\tau_1,\ldots,\tau_K$ from our list such that $T_d(\{\tau_k\}) = \frac{1}{K} \sum_{k=1}^K T_d(\tau_k)$ is close to the true moments $T_d(\{\theta^k\})$ for $d=2,3$. This is a linear system subject to a $K$-sparse constraint, which could perhaps be solved using standard methods such as $\ell_1$-minimization.

\subsection{Summary of techniques}
\label{sec:techniques}

Our approach will draw inspiration from prior work on random overcomplete third-order tensor decomposition. This is the problem of recovering $\{a_1,\ldots,a_r\}$ from
\begin{equation}\label{eq:T-decomp}
T = \sum_{i=1}^r a_i^{\otimes 3}
\end{equation}
where the $a_i \in \mathbb{R}^p$ are drawn independently from $\mathcal{N}(0,I/p)$. The state-of-the-art results for this problem are a close-to-linear-time spectral method that succeeds when $r \lesssim p^{4/3}$ \cite{HSSS-fast-sos} and a polynomial-time sum-of-squares method that succeeds when $r \lesssim p^{3/2}$ \cite{MSS-tensor-sos}. (It seems likely that no efficient algorithm can succeed when $r$ exceeds $p^{3/2}$.)

As a starting point for our techniques, we consider the spectral method of \cite{HSSS-fast-sos} for random overcomplete tensor decomposition. The key step of the algorithm is to construct (from $T$) the $p^2 \times p^2$ matrix
\begin{equation}\label{eq:M-hsss}
M = \sum_{i,j \in [r]} \langle u, \tilde T (a_i \otimes a_j) \rangle \cdot (a_i \otimes a_j)(a_i \otimes a_j)^\top
\end{equation}
where $u \in \mathbb{R}^p$ is drawn randomly from $\mathcal{N}(0,I)$ and $\tilde T$ is obtained by flattening the input tensor to a $p \times p^2$ matrix: $\tilde T = \sum_{i \in [r]} a_i(a_i \otimes a_i)^\top$. The idea is that with some decent probability (inverse polynomial), the random vector $u$ will align reasonably well with some $a_i$, and this causes the top eigenvector of $M$ (after applying a certain ``preconditioner'') to be close to $a_i \otimes a_i$.

We can re-interpret the matrix $M$ in the graphical language of tensor networks (see e.g.\ \cite{tensor-net}), which we now describe. An order-$d$ tensor $T \in (\mathbb{R}^p)^{\otimes d}$ is represented graphically as having $d$ \emph{legs}; the case $d=3$ is shown in Figure~\ref{fig:net-intro}(a). The legs are labeled with the three indices $a,b,c$ that index into $T$. When two tensor legs are connected by a \emph{wire}, this indicates contraction of the corresponding indices. For instance, the tensor network in Figure~\ref{fig:net-intro}(b) represents the tensor $B \in (\mathbb{R}^{p})^{\otimes 4}$ given by $B_{abcd} = \sum_{i \in [p]} T_{abi}T_{cdi}$. The matrix $M$ from (\ref{eq:M-hsss}) is the $(\{a,b\},\{c,d\})$-flattening of the tensor $C \in (\mathbb{R}^p)^{\otimes 4}$ that is represented by Figure~\ref{fig:net-intro}(c). Specifically,
\begin{equation}\label{eq:M-net}
M_{ab,cd} = C_{abcd} = \sum_{i,j,k \in [p]} T_{acj}T_{bdk}T_{ijk}u_i.
\end{equation}
One can check that (\ref{eq:M-net}) is equivalent to (\ref{eq:M-hsss}) when $T$ is given by (\ref{eq:T-decomp}).

Now that we have expressed the matrix $M$ from \cite{HSSS-fast-sos} as a tensor network, this opens the door to exploring a whole class of new spectral methods obtained by building various tensor networks out of the input tensor $T$. For instance, for the continuous MRA problem we will see that the tensor network in Figure~\ref{fig:net-intro}(c) does not work but that a larger one, shown in Figure~\ref{fig:net-9}, does. In Appendix~\ref{app:design} we explain in detail some of the considerations involved in choosing this particular tensor network.

% algorithm
We now describe our algorithm in more detail. Similarly to \cite{HSSS-fast-sos}, our algorithm takes in a random guess $u$ in order to break symmetry. Instead of a vector, $u$ is now an order-5 tensor (with i.i.d.\ $\mathcal{N}(0,1)$ entries). (In Appendix~\ref{app:design} we explain the reason for this.) The hope is that $u$ has better-than-random correlation with $\theta^{\otimes 5}$ for some $\theta$ in the orbit of one of the true signals $\theta^1,\ldots,\theta^K$; if this occurs then we will recover a vector close to $\theta$. Our algorithm takes $u$ and the input tensor $\mathcal{T}$, and constructs a $p^2 \times p^2$ matrix $\tilde M(\mathcal{T},u)$ according to the tensor network in Figure~\ref{fig:net-9}. We would like it to be the case that if we correctly guess $u = \theta^{\otimes 5}$ then $\tilde M(\mathcal{T},\theta^{\otimes 5}) \approx (\theta^{\otimes 2})(\theta^{\otimes 2})^\top$, allowing us to recover $\theta$. Due to the combinatorics of the $\mathrm{SO}(2)$ structure, this is not the case for $\tilde M$; however, luckily it is true after applying a particular simple correction to $\tilde M$, resulting in a matrix $M(\mathcal{T},u)$. (This correction operates entrywise in the Fourier basis.) To extract a candidate solution from $M$, we symmetrize it and compute its top eigenvector $v \in \mathbb{R}^{p^2}$ (which we hope is close to $\theta^{\otimes 2}$). We then re-shape $v$ into a $p \times p$ matrix, symmetrize it, and take the top eigenvector again in order to produce a candidate solution. We then repeat the entire process $L$ times with fresh randomness $u$ on each trial, in order to obtain a list of $L$ candidate vectors.

% trace method
Roughly speaking, a key step in our analysis is to show a high-probability upper bound on the spectral norm of our matrix $M = M(\mathcal{T},u)$. To do this we use the trace moment method, a general-purpose tool from random matrix theory which relies on computing
\begin{equation}\label{eq:trace}
\mathbb{E}[\mathrm{Tr}((MM^\top)^q)].
\end{equation}
In general, this computation can be quite difficult for complicated random matrices. However, even though $M$ is quite complicated, the fact that it is represented by a tensor network helps us here. As shown in Figure~\ref{fig:trace-method}, a tensor network for the quantity (\ref{eq:trace}) can be obtained by connecting $2q$ copies of $M$ in a circle. Since $M$ is itself a tensor network, we need to connect $2q$ copies of \emph{that} network in a circle, creating an expanded tensor network. As a result, the computation of (\ref{eq:trace}) boils down to a combinatorics question involving counting certain labelings of this expanded tensor network.

% figure: trace method intro
\begin{figure}[!ht]\centering
\begin{minipage}{0.8\textwidth}\centering
\includegraphics[width=0.5\textwidth]{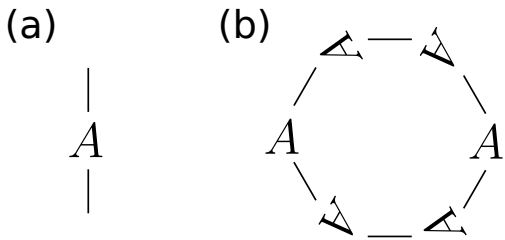}
\caption{{\bf (a)} A real-valued rectangular matrix $A$. {\bf (b)} The tensor network representation of $\mathrm{Tr}[(AA^\top)^q]$ is formed by connecting $2q$ copies of $A$ in a ring (here $q=3$). Since $A$ is asymmetric, the orientation of the ``$A$'' symbols matters.}
\label{fig:trace-method}
\end{minipage}
\end{figure}

%%% continuous MRA proof
\section{Proof for continuous MRA}

\subsection{Preliminaries}

\subsubsection{Concentration}

First we have some basic concentration results for random vectors.

\begin{lemma}
If $\theta \sim \mathcal{N}(0,I_p/p)$ then 
$$\left|\|\theta\|^2 - 1\right| \le \tilde O(1/\sqrt{p})$$
with overwhelming probability.
\end{lemma}
\begin{proof}
This follows from Bernstein's inequality for subexponential random variables (see e.g.\ \cite{rig-notes}).
\end{proof}

\begin{lemma}
If $\theta \sim \mathcal{N}(0,I_p/p)$ then with overwhelming probability we have for all $i$, $$|\theta_i| \le \tilde O(1/\sqrt{p}).$$
\end{lemma}
\begin{proof}
This follows from standard Gaussian tail bounds.
\end{proof}

The following concentration bound is a consequence of \emph{hypercontractivity} (see e.g.\ Theorem~1.10 of \cite{hyp}).

\begin{theorem}\label{thm:hyp}
Consider a degree-$q$ polynomial $f(Y) = f(Y_1,\ldots,Y_n)$ of independent Gaussian random variables $Y_1,\ldots,Y_n$. Let $\sigma^2$ be the variance of $f(Y)$. There exists an absolute constant $R > 0$ such that
$$\mathrm{Pr}\left[\left|f(Y) - \mathbb{E}[f(Y)]\right| \ge t\right] \le e^2 \cdot e^{-\left(\frac{t^2}{R \sigma^2}\right)^{1/q}}.$$
\end{theorem}

\subsubsection{Fourier basis}\label{sec:fourier}

% change of basis Delta; swap; dotted line
We will largely work in the Fourier domain. Let $\Delta$ be the unitary matrix that converts from the Fourier representation to the standard representation of a vector $v \in \mathbb{R}^p$, i.e.\ $\theta = \Delta \hat\theta$; see (\ref{eq:Delta}). We define the Fourier transform $\hat T$ of a tensor $T$ as depicted in Figure~\ref{fig:change-basis}(b). One can check that $\Delta^\top \Delta$ is the permutation matrix that swaps indices $i$ and $-i$, i.e.\ $(\Delta^\top \Delta)_{ij} = \mathbbm{1}_{i=-j}$. Thus, as shown in Figure~\ref{fig:change-basis}(c), when two copies of $\Delta$ combine in a tensor network, we denote this by a dotted line which is understood to mean contraction with indices $i$ and $-i$ paired.

% figure: change of basis Delta
\begin{figure}[!ht]\centering
\begin{minipage}{0.8\textwidth}\centering
\includegraphics[width=0.8\textwidth]{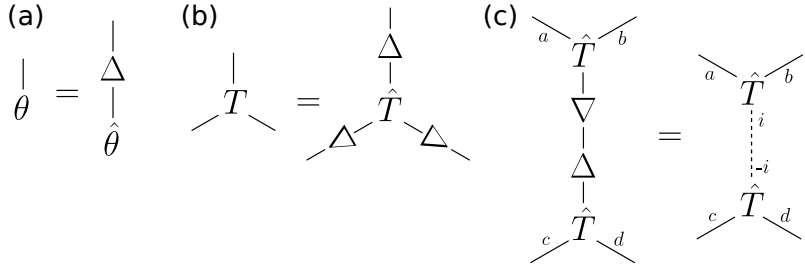}
\caption{{\bf (a)} The matrix $\Delta$ converts a vector's Fourier representation $\hat\theta$ to its standard representation $\theta$. {\bf (b)} We can convert from $\hat T$ to $T$ by attaching three copies of $\Delta$. Note that $\Delta$ is asymmetric and so the orientation of the $\Delta$ symbols is important. {\bf (c)} When two $\Delta$'s connect as shown, this has the effect of a contraction in which indices $i$ and $-i$ are paired. We abbreviate this as a dotted line with one end labeled $i$ and the other end labeled $-i$. The tensor $C$ shown here is $C_{abcd} = \sum_i \hat T_{abi} \hat T_{cd(-i)}$.}
\label{fig:change-basis}
\end{minipage}
\end{figure}

\subsection{Main technical theorem}

We now begin the proof of our main result (Theorem~\ref{thm:mra-main}).

Our algorithm will build its list of candidate solutions by repeating a certain spectral method $L$ times, with fresh randomness $u$ each time. The following main technical theorem shows that each of these trials has a decent probability of success.

\begin{theorem}[main technical theorem]
\label{thm:mra-tech}
Let $\{\theta^k\}$, $\delta$ and $\mathcal{T}$ be as in Theorem~\ref{thm:mra-main}. Let $K \le p^\delta$. Let $u \in (\mathbb{R}^p)^{\otimes 5}$ be drawn from $\mathcal{N}(0,I_{p^5})$. There is a matrix $M(\mathcal{T},u) \in \mathbb{R}^{p^2 \times p^2}$ (computable in time $\mathrm{poly}(p)$ from $\mathcal{T}$ and $u$) with the following guarantee. Let $v \in \mathbb{R}^{p^2}$ be the leading eigenvector of $\frac{1}{2}[M(\mathcal{T},u)+M(\mathcal{T},u)^\top]$. Re-shape\footnote{We will see that $M(\mathcal{T},u)$ is a flattening of a 4-tensor, with entries $M(\mathcal{T},u)_{ab,cd}$. Thus $v$ has entries $v_{ab}$ and can be naturally thought of as a $p \times p$ matrix.} $v$ to a $p \times p$ matrix $V$ and let $\tau \in \mathbb{R}^p$ be the (unit-norm) leading eigenvector\footnote{Here the \emph{leading eigenvector} is defined to be the one whose eigenvalue is largest in absolute value.} of $\frac{1}{2}(V+V^\top)$. There is a deterministic predicate $P(\{\theta^k\})$ (defined in Section~\ref{sec:together}) depending only on $\{\theta^k\}$, that is satisfied with high probability (over $\{\theta^k\}$). For fixed $\{\theta^k\}$ satisfying $P(\{\theta^k\})$, for any $k \in [K]$ and any $\varepsilon > 0$, we have $\langle \tau,\theta^k \rangle^2 \ge 1 - \varepsilon - o(1)$ with probability $p^{-O(1)/\varepsilon^4}$ over the randomness of $u$.
\end{theorem}

We first see how our main technical theorem (Theorem~\ref{thm:mra-tech}) implies our main theorem (Theorem~\ref{thm:mra-main}).

\begin{proof}[Proof of Theorem~\ref{thm:mra-main}]
To produce the list $\tau_1,\ldots,\tau_L$, the algorithm draws independent samples $u_1,\ldots,u_L \sim \mathcal{N}(0,I_{p^5})$. For $i \in [L]$, extract $\tau_i$ from $M(\mathcal{T},u_i)$ as in Theorem~\ref{thm:mra-tech}. For fixed $k$, let $\gamma = p^{-O(1)/\varepsilon^4}$ denote the success probability of a single trial. For fixed $k$, the probability of success after $L$ trials is at least $1-(1-\gamma)^L \ge 1-\exp(-\gamma L)$. Taking a union bound over $[K]$, the
overall probability of failure is at most $K \exp(-\gamma L) \le p^\delta \exp(-\gamma L)$. To make this $o(1)$, it is sufficient to take $L = \log^2(p)/\gamma = p^{O(1)/\varepsilon^4}$.
\end{proof}

We now begin the proof of the main technical theorem (Theorem~\ref{thm:mra-tech}). The $p^2 \times p^2$ matrix $M(\mathcal{T},u)$ is the $(\{a,b\},\{c,d\})$-flattening of the tensor depicted in Figure~\ref{fig:net-9}, but with an additional post-processing operator $\mathcal{S}$ applied to it. This operator is easiest to describe in the Fourier domain: let $\hat{\mathcal{S}}$ be the operator that acts entrywise on a 4-tensor by multiplying the $abcd$ entry by a nonnegative real number $S_{abcd}$ to be specified later. We will have $S_{abcd} = S_{(-a)(-b)(-c)(-d)}$ and so $\mathcal{S}$ takes real 4-tensors to real 4-tensors. We define
$$M(\mathcal{T},u) = (\Delta \otimes \Delta)[\hat M(\mathcal{T},u)](\Delta \otimes \Delta)^\top$$
where $\Delta$ is as in Section~\ref{sec:fourier}, and where $\hat M(\mathcal{T},u)$ is obtained by applying $\hat{\mathcal{S}}$ to the $(\{a,b\},\{c,d\})$-flattening of the tensor depicted in Figure~\ref{fig:net-mra}.

% figure: M(T,u)
\begin{figure}\centering
\begin{minipage}{0.8\textwidth}\centering
\begin{tikzpicture}[auto]
\def\ra{2};  % inner radius
\def\rad{3};  % outer radius

\node (u) at (0,0) {$\hat u$};
\node (T1) at ({-\ra*cos(1*360/9)},{\ra*sin(1*360/9)}) {$\hat{\mathcal{T}}$};
\node (T2) at ({-\ra*cos(2*360/9)},{\ra*sin(2*360/9)}) {$\hat{\mathcal{T}}$};
\node (T3) at ({-\ra*cos(3*360/9)},{\ra*sin(3*360/9)}) {$\hat{\mathcal{T}}$};
\node (T4) at ({-\ra*cos(4*360/9)},{\ra*sin(4*360/9)}) {$\hat{\mathcal{T}}$};
\node (T5) at ({-\ra*cos(5*360/9)},{\ra*sin(5*360/9)}) {$\hat{\mathcal{T}}$};
\node (T6) at ({-\ra*cos(6*360/9)},{\ra*sin(6*360/9)}) {$\hat{\mathcal{T}}$};
\node (T7) at ({-\ra*cos(7*360/9)},{\ra*sin(7*360/9)}) {$\hat{\mathcal{T}}$};
\node (T8) at ({-\ra*cos(8*360/9)},{\ra*sin(8*360/9)}) {$\hat{\mathcal{T}}$};
\node (T9) at ({-\ra*cos(9*360/9)},{\ra*sin(9*360/9)}) {$\hat{\mathcal{T}}$};
\coordinate (a) at ({-\rad*cos(1*360/9)},{\rad*sin(1*360/9)});
\coordinate (c) at ({-\rad*cos(3*360/9)},{\rad*sin(3*360/9)});
\coordinate (b) at ({-\rad*cos(5*360/9)},{\rad*sin(5*360/9)});
\coordinate (d) at ({-\rad*cos(7*360/9)},{\rad*sin(7*360/9)});
\draw[dashed] (T1) to node[very near start,inner sep=1pt] {\scriptsize{$i_2$}} node[very near end,inner sep=1pt] {\scriptsize{-$i_2$}} (T2);
\draw[dashed] (T2) to node[pos=0,inner sep=1pt] {\scriptsize{$i_3$}} node[pos=0.6,inner sep=1pt] {\scriptsize{-$i_3$}} (T3);
\draw[dashed] (T3) to node[very near start,inner sep=1pt] {\scriptsize{$i_4$}} node[very near end,inner sep=1pt] {\scriptsize{-$i_4$}} (T4);
\draw[dashed] (T4) to node[very near start,inner sep=2pt] {\scriptsize{$i_5$}} node[very near end,inner sep=2pt] {\scriptsize{-$i_5$}} (T5);
\draw[dashed] (T5) to node[very near start,inner sep=1pt] {\scriptsize{$i_6$}} node[very near end,inner sep=1pt] {\scriptsize{-$i_6$}} (T6);
\draw[dashed] (T6) to node[pos=0.3,inner sep=1pt] {\scriptsize{$i_7$}} node[pos=1,inner sep=1pt] {\scriptsize{-$i_7$}} (T7);
\draw[dashed] (T7) to node[very near start,inner sep=1pt] {\scriptsize{$i_8$}} node[very near end,inner sep=1pt] {\scriptsize{-$i_8$}} (T8);
\draw[dashed] (T8) to node[very near start,inner sep=1pt] {\scriptsize{$i_9$}} node[very near end,inner sep=1pt] {\scriptsize{-$i_9$}} (T9);
\draw[dashed] (T9) to node[very near start,inner sep=1pt] {\scriptsize{$i_1$}} node[very near end,inner sep=1pt] {\scriptsize{-$i_1$}} (T1);
\draw (T1) -- (a) node [midway,above,inner sep=2pt] {\scriptsize{$a$}};
\draw[dashed] (T2) to node[near start,inner sep=1pt] {\scriptsize{$j_1$}} node[very near end,inner sep=1pt] {\scriptsize{-$j_1$}} (u);
\draw (T3) -- (c) node [midway,right,inner sep=2pt] {\scriptsize{$c$}};
\draw[dashed] (T4) to node[near start,inner sep=1pt] {\scriptsize{$j_2$}} node[very near end,inner sep=1pt] {\scriptsize{-$j_2$}} (u);
\draw (T5) -- (b) node [midway,below,inner sep=2pt] {\scriptsize{$b$}};
\draw[dashed] (T6) to node[near start,inner sep=1pt] {\scriptsize{$j_3$}} node[very near end,inner sep=1pt] {\scriptsize{-$j_3$}} (u);
\draw (T7) -- (d) node [midway,left,inner sep=2pt] {\scriptsize{$d$}};
\draw[dashed] (T8) to node[near start,inner sep=1pt] {\scriptsize{$j_4$}} node[pos=0.7,inner sep=1pt] {\scriptsize{-$j_4$}} (u);
\draw[dashed] (T9) to node[near start,inner sep=1pt] {\scriptsize{$j_5$}} node[very near end,inner sep=1pt] {\scriptsize{-$j_5$}} (u);
\end{tikzpicture}
\caption{The $p^2 \times p^2$ matrix $\hat M(\mathcal{T},u)$ is obtained by applying $\hat{\mathcal{S}}$ to the $(\{a,b\},\{c,d\})$-flattening of the tensor shown here. The dotted lines and the Fourier transforms $\hat{\mathcal{T}}$ and $\hat u$ are defined as in Figure~\ref{fig:change-basis}.}
\label{fig:net-mra}
\end{minipage}
\end{figure}

% write out formula for \hat M(T,u)
Explicitly, we have
\begin{align*}
\hat M(\mathcal{T},u)_{ab,cd} = S_{abcd} \sum_{i_1,\ldots,i_9} \sum_{j_1,\ldots,j_5} \hat u_{-j_1, -j_2, -j_3, -j_4, -j_5}
\hat{\mathcal{T}}_{-i_1,a,i_2}
&\hat{\mathcal{T}}_{-i_2,j_1,i_3}
\hat{\mathcal{T}}_{-i_3,c,i_4}
\hat{\mathcal{T}}_{-i_4,j_2,i_5}
\hat{\mathcal{T}}_{-i_5,b,i_6}\\
&\hat{\mathcal{T}}_{-i_6,j_3,i_7}
\hat{\mathcal{T}}_{-i_7,d,i_8}
\hat{\mathcal{T}}_{-i_8,j_4,i_9}
\hat{\mathcal{T}}_{-i_9,j_5,i_1}.
\end{align*}

% breakdown into terms
Recall $\mathcal{T} = T + E$ where $$T = \sum_{k=1}^K \,\Ex_g \left[(g \cdot \theta^k)^{\otimes 3}\right].$$
Let $\theta^k$ be the signal we are hoping to recover. Let
$$T^k = \Ex_g \left[(g \cdot \theta^k)^{\otimes 3}\right].$$

Recall $u \sim \mathcal{N}(0,I_{p^5}) \in (\mathbb{R}^p)^{\otimes 5}$. Write $u = \alpha \,(\theta^k)^{\otimes 5} + \tilde u$ with $\tilde u \perp (\theta^k)^{\otimes 5}$.
We will break down the matrix $M(\mathcal{T},u)$ into the following terms:
\begin{align*}
M(\mathcal{T},u) &= M(T,u) + [M(\mathcal{T},u) - M(T,u)]\\
&= \alpha M(T,(\theta^k)^{\otimes 5}) + M(T,\tilde u) + [M(\mathcal{T},u) - M(T,u)]\\
&= \alpha M(T^k,(\theta^k)^{\otimes 5}) + \alpha[M(T,(\theta^k)^{\otimes 5}) - M(T^k,(\theta^k)^{\otimes 5})] + M(T,\tilde u) + [M(\mathcal{T},u) - M(T,u)].
\end{align*}
Here we have used the fact that $M(\mathcal{T},u)$ is linear in $u$. We now have four terms to bound separately.

\subsection{Signal term}\label{sec:signal}

Here we consider the signal term $M(T^k,(\theta^k)^{\otimes 5})$. Let $\Theta^k = (\theta^k \otimes \theta^k)(\theta^k \otimes \theta^k)^\top$, the matrix we would like to recover. Intuitively, we will show that if we were to correctly guess $u = (\theta^k)^{\otimes 5}$, then the resulting matrix matrix would be close to $\Theta^k$. In order for this to be true, we will need to choose the parameters $S_{abcd}$ appropriately.

% prop: signal term close to rank-1
\begin{proposition} \label{prop:signal}
For any $k \in [K]$, with overwhelming probability over $\theta^k$,
$$\|M(T^k,(\theta^k)^{\otimes 5}) - \Theta^k\| \le o(1).$$
\end{proposition}

\noindent This section is devoted to proving Proposition~\ref{prop:signal}. Recall
$$\hat T_{i_1 i_2 i_3} = \mathbbm{1}_{i_1+i_2+i_3=0}\sum_{k=1}^K \hat\theta^k_{i_1} \hat\theta^k_{i_2} \hat\theta^k_{i_3}$$
and so
$$\hat T^k_{i_1 i_2 i_3} = \mathbbm{1}_{i_1+i_2+i_3=0}\, \hat\theta^k_{i_1} \hat\theta^k_{i_2} \hat\theta^k_{i_3}.$$
Without loss of generality, take $k=1$. We have
\begin{equation}\label{eq:hatM1}
\hat M(T^1,(\theta^1)^{\otimes 5})_{ab,cd} = S_{abcd}\, s_{abcd}\, \hat\theta^1_a \hat\theta^1_b \hat\theta^1_c \hat\theta^1_d
\end{equation}
where
$$s_{abcd} = \sum_{i_1,\ldots,i_9} \sum_{j_1,\ldots,j_5} \left(\mathbbm{1}_{-i_1+a+i_2=0} \cdots \mathbbm{1}_{-i_9+j_5+i_1=0}\right) \left( |\hat\theta^1_{i_1}|^2 \cdots |\hat\theta^1_{i_9}|^2 |\hat\theta^1_{j_1}|^2 \cdots |\hat\theta^1_{j_5}|^2 \right).$$
Here the indicator functions enforce that for each copy of $\hat{\mathcal{T}}$ in Figure~\ref{fig:net-mra}, the three incident labels sum to zero. Define
$$S_{abcd} \defeq \left\{\begin{array}{ll} 0 & \text{if }a=-b\text{ or }c=-d\\ 1/\mathbb{E}[s_{abcd}] & \text{otherwise}. \end{array}\right.$$
The reason for zeroing out some $S_{abcd}$'s will not be apparent until later (Section~\ref{sec:noise}); this is crucially used in the proof of Lemma~\ref{lem:count-labelings} for bounding the noise term $M(T,\tilde u)$. The reason for $1/\mathbb{E}[s_{abcd}]$ should be clear from (\ref{eq:hatM1}).

We will show that $s_{abcd}$ concentrates near its expectation. We start with a basic computation of the moments of $\hat\theta^1$ (which of course holds for any $\hat\theta^k$).
\begin{lemma}\label{lem:moments}
$\mathbb{E}|\hat\theta^1_i|^{2k} = k!\, p^{-k}$. If $k_1 \ne k_2$ then $\mathbb{E}[(\hat\theta_i^1)^{k_1} (\hat\theta_{-i}^1)^{k_2}] = 0$. If $i \ne \pm j$ then $\hat\theta_i^1$ and $\hat\theta_j^1$ are independent.
\end{lemma}
\begin{proof}
The third statement is immediate from (\ref{eq:Delta}), since $\theta^1 \sim \mathcal{N}(0,I/p)$. The second statement is immediate from the fact that the complex phase of $\hat\theta_i^1$ is a uniformly random angle, and $\hat\theta_{-i}^1 = \overline{\hat\theta_i^1}$. For the first statement, $|\hat\theta^1_i|^2 \sim \frac{1}{2p} \chi_2^2$, so use the known formula for chi-squared moments: $\mathbb{E}[(\chi_2^2)^k] = 2^k k!$. 
\end{proof}

% lemma: E[s_abcd]
We next show that for every $a,b,c,d$ we have $\mathbb{E}[s_{abcd}] = \Theta(p^{-9})$, specifically:
\begin{lemma} \label{lem:Es}
There exist universal positive constants $c_1$ and $c_2$ such that for every $a,b,c,d \in \pm [p/2]$,
$$c_1\, p^{-9} \le \mathbb{E}[s_{abcd}] \le c_2\, p^{-9}.$$
\end{lemma}

\begin{proof}

% upper bound
Fix $a,b,c,d$. There is a (nonzero) term of $s_{abcd}$ for each choice of indices $i_1,\ldots,i_9,j_1,\ldots,j_5 \in \pm [p/2]$ such that for each copy of $\hat{\mathcal{T}}$ in Figure~\ref{fig:net-mra}, the three incident indices sum to zero. There are at most $p^5$ (nonzero) terms in $s_{abcd}$ because once $i_9,j_5,j_4,j_3,j_1$ are chosen, the zero-sum constraints uniquely determine at most one possible value for the other indices. (We say ``at most one'' since only indices in the set $\pm [p/2]$ are valid.) Each term of $s_{abcd}$ has expectation at most $14!\, p^{-14}$ (by Lemma~\ref{lem:moments}), so $\mathbb{E}[s_{abcd}] \le 14!\, p^{-9}$. This proves the upper bound.

% lower bound proof sketch
The idea of the lower bound is to argue that $s_{abcd}$ has $\Omega(p^5)$ terms and each term has expectation at least $p^{-14}$. We defer the full proof to Section~\ref{sec:pf-Es}.
\end{proof}

% Var[s_{abcd}]
\begin{lemma}
There exists a universal positive constant $c_3$ such that for every $a,b,c,d \in \pm [p/2]$,
$$\mathrm{Var}[s_{abcd}] \le c_3\, p^{-19}.$$
\end{lemma}
\begin{proof}
The variance of a sum can be broken down as $\mathrm{Var}(\sum_i x_i) = \sum_i \mathrm{Var}(x_i) + \sum_{i \ne j} \mathrm{Covar}(x_i,x_j)$. Each of the $O(p^5)$ terms of $s_{abcd}$ has variance $O(p^{-28})$. There are $O(p^{10})$ ways to choose two distinct terms of $s_{abcd}$. Only $O(p^9)$ of these ways gives two terms that are dependent, in which case their covariance is $O(p^{-28})$; otherwise they are independent and have covariance zero. This means $\mathrm{Var}[s_{abcd}] \le O(p^5 \cdot p^{-28} + p^9 \cdot p^{-28}) = O(p^{-19})$.
\end{proof}

% hypercontractivity
By hypercontractivity (Theorem~\ref{thm:hyp}) we have with overwhelming probability, $|s_{abcd} - \mathbb{E}[s_{abcd}]| \le p^{-9.1}$. Thus, when $a \ne -b$ and $c \ne -d$, we have $|S_{abcd} \,s_{abcd} - 1| \le O(p^{-0.1})$ (recall that $S_{abcd} s_{abcd}$ appears in (\ref{eq:hatM1})). For the entries with $a=-b$ or $c=-d$ (there are $ \le 2p^3$ such entries in $\hat M$), we will simply use the bound $|\hat\theta^1_i| \le \tilde O(1/\sqrt{p})$.

We can now complete the proof of Proposition~\ref{prop:signal}. Using (\ref{eq:hatM1}),
\begin{align*}
\|M(T^1,(\theta^1)^{\otimes 5}) - (\theta^1 \otimes \theta^1)(\theta^1 \otimes \theta^1)^\top\|
&= \|\hat M(T^1,(\theta^1)^{\otimes 5}) - (\hat\theta^1 \otimes \hat\theta^1)(\hat\theta^1 \otimes \hat\theta^1)^\top\| \\
&\le \|\hat M(T^1,(\theta^1)^{\otimes 5}) - (\hat\theta^1 \otimes \hat\theta^1)(\hat\theta^1 \otimes \hat\theta^1)^\top\|_F \\
&\le \sqrt{p^4 \cdot \tilde O(p^{-0.1} \cdot p^{-2})^2 + 2p^3 \cdot \tilde O(p^{-2})^2} \\
&\le o(1).
\end{align*}

\subsection{Noise term}\label{sec:noise}

We now consider the noise term $M(T,\tilde u)$, i.e.\ the term created by the ``bad'' component of $u$ that is orthogonal to $(\theta^k)^{\otimes 5}$. This term is the crux of the proof, where we will crucially use the assumption $K \le p^\delta$.

\begin{proposition} \label{prop:noise-u}
There exists $\delta > 0$ such that if $K \le p^\delta$, we have the following. There is a determinstic predicate $P_1(\{\theta^k\})$ depending only on $\{\theta^k\}$ that is satisfied with high probability. For any fixed $\{\theta^k\}$ satisfying $P_1(\{\theta^k\})$, we have
$$\|M(T,\tilde u)\| \le O(\sqrt{\log p})$$
with high probability over the randomness of $\tilde u$.
\end{proposition}

\begin{remark}
Note that $\tilde u$ depends on which signal $k \in [K]$ we have chosen as the target, since $\tilde u \perp (\theta^k)^{\otimes 5}$. However, $P_1$ does not depend on $k$, and the conclusion of Proposition~\ref{prop:noise-u} holds for any fixed $k$.
\end{remark}

This section is devoted to proving Proposition~\ref{prop:noise-u}. We will use the following result on random contractions of tensors.

% norm of random contraction -- from [MSS]
\begin{theorem}[\cite{MSS-tensor-sos} Corollary~6.6]
\label{thm:MSS-rand-contract}
Let $W \in \mathbb{R}^p \times \mathbb{R}^q \times \mathbb{R}^r$ be an order-3 tensor. Let $\tilde u \sim \mathcal{N}(0,\Sigma)$ with $r \times r$ covariance matrix satisfying $0 \preceq \Sigma \preceq I$. Then for any $t \ge 0$,
$$\Pr_{\tilde u}\left[\|(I \otimes I \otimes \tilde u^\top)W\|_{\{1\},\{2\}} \ge t \cdot \max\left\{\|W\|_{\{1\},\{2,3\}},\|W\|_{\{1,3\},\{2\}}\right\}\right] \le 4(p+q) \exp(-t^2/2).$$
\end{theorem}

In our setting, we have $M(T,\tilde u) = (I \otimes I \otimes \tilde u^\top)W$ where $W$ is given by the tensor network in Figure~\ref{fig:mra-trace}(a) ($\hat{\mathcal{S}}$ is present but not shown). Explicitly, the Fourier transform (defined as in Figure~\ref{fig:change-basis}) of $W$ is
$$\hat W_{ab,cd,j_1 j_2 j_3 j_4 j_5} = S_{abcd} \sum_{i_1,\ldots,i_9} \hat T_{-i_1,a,i_2} \hat T_{-i_2,j_1,i_3} \hat T_{-i_3,c,i_4} \hat T_{-i_4,j_2,i_5} \hat T_{-i_5,b,i_6} \hat T_{-i_6,j_3,i_7} \hat T_{-i_7,d,i_8} \hat T_{-i_8,j_4,i_9} \hat T_{-i_9,j_5,i_1}.$$
By Theorem~\ref{thm:MSS-rand-contract},
\begin{equation} \label{eq:rand-contraction}
\Pr_{\tilde u}\left[\|M(T,\tilde u)\| \ge t \sigma \right] \le 8p^2 \exp(-t^2/2)
\end{equation}
where
$$\sigma = \max \{\|W\|_{\{a,b\},\{c,d, j_1, j_2, j_3, j_4, j_5\}},\|W\|_{\{a,b, j_1, j_2, j_3, j_4, j_5\},\{c,d\}}\}.$$

The two flattenings of $W$ that appear in the definition of $\sigma$ are equivalent due to symmetry, so it suffices to consider just the first one. The corresponding matrix is $\tilde W \in \mathbb{R}^{2 \times 7}$ given by
$$\tilde W_{ab,cd j_1 j_2 j_3 j_4 j_5} = W_{ab,cd, j_1 j_2 j_3 j_4 j_5}.$$

We will bound $\|\tilde W\|$ using the trace moment method:
\begin{theorem}[e.g.\ \cite{PS-tensor-completion} Proposition~5.2] \label{thm:trace-method}
For any real-valued random matrix $Y$, for any integer $q \ge 1$ and any $\varepsilon > 0$,
$$\mathrm{Pr}\left[\|Y\| > \left(\frac{\mathbb{E}[\mathrm{Tr}((YY^\top)^q)]}{\varepsilon}\right)^{\frac{1}{2q}}\right] < \varepsilon.$$
\end{theorem}

As illustrated in Figures~\ref{fig:trace-method} and \ref{fig:mra-trace}, we can represent $\mathrm{Tr}((\tilde W \tilde W^\top)^q)$ as a tensor network by connecting (in a ring) $2q$ copies of the tensor network for $\tilde W$. (We will take $q=\log p$.) Call this new tensor network $\mathcal{G}_q$ (see Figure~\ref{fig:mra-trace}). The computation of $\mathbb{E}[\mathrm{Tr}((\tilde W \tilde W^\top)^q)]$ is thus reduced to a combinatorial problem involving labelings of $\mathcal{G}_q$, which we next describe.

% figure: trace method mra
\begin{figure}\centering
\begin{minipage}{0.8\textwidth}\centering
\begin{tikzpicture}[auto]
\node at (-6,3.5) {\bf (a)};
\def\rin{1}
\def\rd{1.75}
\def\rT{2.75}
\def\rD{3.75}
\def\rout{4.5}
% inner j's
\foreach \i in {2,3,4,6,7,8,9}{
	\coordinate (pt\i) at ({180-\i*360/9}:\rin);
}
% inner d's
\foreach \i in {2,3,4,6,7,8,9}{
	\node[rotate={270-\i*360/9}] (d\i) at ({180-\i*360/9}:\rd) {$\Delta$};
}
% T's
\foreach \i in {1,...,9}{
	\node (T\i) at ({180-\i*360/9}:\rT) {$\hat{T}$};
}
% outer D's
\foreach \i in {1,5}{
	\node[rotate={90-\i*360/9}] (d\i) at ({180-\i*360/9}:\rD) {$\Delta$};
}
% outer ab
\foreach \i in {1,5}{
	\coordinate (pt\i) at ({180-\i*360/9}:\rout);
}
% ring
\foreach \i in {1,...,8}{
	\pgfmathtruncatemacro{\nexti}{\i + 1}
	\draw[dashed] (T\i) -- (T\nexti);
}
\draw[dashed] (T9) -- (T1);
% d,D edges
\foreach \i in {1,...,9}{
	\draw (T\i) -- (d\i);
}
% in,out edges
\foreach \i/\a in {1/a,2/j_1,3/c,4/j_2,5/b,6/j_3,7/d,8/j_4,9/j_5}{
	\draw (d\i) to node[midway,inner sep=1pt] {\scriptsize{$\a$}} (pt\i);
}
\begin{scope}[shift={(0,-9)}]
\node at (-6,5.5) {\bf (b)};
\def\rin{0.75}
\def\rt{1.5}
\def\rtt{3}
\def\rttt{5}
\def\rout{6.5}
% inner coordinates
\foreach \i in {2,3,4,6,7,8,9}{
	\coordinate (pt\i) at ({180-\i*360/9}:\rin);
}
% T's
\foreach \i in {1,...,9}{
	\node (t\i) at ({180-\i*360/9}:\rt) {$\hat{T}$};
	\node (tt\i) at ({180-\i*360/9}:\rtt) {$\hat{T}$};
	\node (ttt\i) at ({180-\i*360/9}:\rttt) {$\hat{T}$};
}
% outer coordinates
\foreach \i in {1,5}{
	\coordinate (pt\i) at ({180-\i*360/9}:\rout);
}
% rings
\foreach \i in {1,...,8}{
	\pgfmathtruncatemacro{\nexti}{\i + 1}
	\draw[dashed] (t\i) -- (t\nexti);
	\draw[dashed] (tt\i) to node[near start,inner sep=1pt] {\scriptsize{$\tilde i_\nexti^{(1)}$}} node[near end,inner sep=1pt] {\scriptsize{-$\tilde i_\nexti^{(1)}$}} (tt\nexti);
	\draw[dashed] (ttt\i) to node[near start,inner sep=1pt] {\scriptsize{$i_\nexti^{(1)}$}} node[near end,inner sep=1pt] {\scriptsize{-$i_\nexti^{(1)}$}} (ttt\nexti);
}
\draw[dashed] (t9) -- (t1);
\draw[dashed] (tt9) to node[near start,inner sep=1pt] {\scriptsize{$\tilde i_1^{(1)}$}} node[near end,inner sep=1pt] {\scriptsize{-$\tilde i_1^{(1)}$}} (tt1);
\draw[dashed] (ttt9) to node[near start,inner sep=1pt] {\scriptsize{$i_1^{(1)}$}} node[near end,inner sep=1pt] {\scriptsize{-$i_1^{(1)}$}} (ttt1);
% j's
\foreach \i/\j in {2/j_1,3/c,4/j_2,6/j_3,7/d,8/j_4,9/j_5}{
	\draw[dashed] (ttt\i) to node[near start,inner sep=1pt] {\scriptsize{$\j^{(1)}$}} node[near end,inner sep=1pt] {\scriptsize{-$\j^{(1)}$}} (tt\i);
	\draw[dashed] (t\i) -- (pt\i);
}
% ab's
\foreach \i/\a in {1/a,5/b}{
	\draw[dashed] (pt\i) to node[near start,inner sep=1pt] {\scriptsize{-$\a^{(1)}$}} node[near end,inner sep=1pt] {\scriptsize{$\a^{(1)}$}} (ttt\i);
	\draw[dashed] (tt\i) to node[near start,inner sep=1pt] {\scriptsize{-$\a^{(2)}$}} node[near end,inner sep=1pt] {\scriptsize{$\a^{(2)}$}} (t\i);
}
\end{scope}
\end{tikzpicture}
\caption{{\bf (a)} The tensor network for $W$. (The operator $\hat S$ is not shown but can be thought of as living on the appropriate four edges.) The matrix $\tilde W$ is the $(\{a,b\},\{c,d,j_1,j_2,j_3,j_4,j_5\})$-flattening. {\bf (b)} Here we see (part of) the tensor network $\mathcal{G}_q$ that computes $\mathrm{Tr}((\tilde W \tilde W^\top)^q)$. There are $2q$ copies of the tensor network from (a) (3 copies are shown here) connected in a ring as in Figure~\ref{fig:trace-method}. The outermost copy connects back to the innermost copy, so that the entire network can be visualized as living on the surface of a torus. Again, $\hat{\mathcal{S}}$ is not shown (it will be unimportant since we will bound its contribution separately). Recall that connecting two ``opposing'' copies of $\Delta$ results in a dotted edge, as shown in Figure~\ref{fig:change-basis}.}
\label{fig:mra-trace}
\end{minipage}
\end{figure}

\begin{definition}
A \emph{labeling} $\mathcal{L}$ of $\mathcal{G}_q$ is described by the following. For each edge $e$, label one end with a value $i_e \in \pm [p/2]$ and label the other end with $-i_e$. Call each copy of $\hat T$ in $\mathcal{G}_q$ a \emph{vertex}, and label each vertex $v$ with a value $k_v \in [K]$. Let $\mathcal{L}(v) = \mathbbm{1}_{i_1+i_2+i_3=0} \,\hat\theta^{k_v}_{i_1} \hat\theta^{k_v}_{i_2} \hat\theta^{k_v}_{i_3}$ where $i_1,i_2,i_3$ are the three edge labels incident\footnote{Suppose an edge $e$ is incident to a vertex $v$ in a tensor network. There is a label $\pm i_e$ at each end of $e$. The label at the $v$-end of $e$ is considered \emph{incident} to $v$.} to $v$.
\end{definition}

\noindent Recall that $\hat T_{i_1 i_2 i_3} = \mathbbm{1}_{i_1+i_2+i_3=0}\sum_{k=1}^K \hat\theta^k_{i_1} \hat\theta^k_{i_2} \hat\theta^k_{i_3}$. The vertex labels $k_v$ correspond to the terms in this sum.

As shown in Figure~\ref{fig:mra-trace}, there are $q$ \emph{layers}, and layer $\ell$ (for $\ell = 1,2,\ldots,q$) has labels $$a^{(\ell)},b^{(\ell)},c^{(\ell)},d^{(\ell)},i_1^{(\ell)},\ldots,i_9^{(\ell)},\tilde i_1^{(\ell)},\ldots,\tilde i_9^{(\ell)}, j_1^{(\ell)},\ldots,j_5^{(\ell)}.$$
Layer numbers are defined modulo $q$, i.e.\ layer $q+1$ refers to layer 1.

We now have
$$\mathbb{E}[\mathrm{Tr}((\tilde W \tilde W^\top)^q)]
= \mathbb{E}\sum_\mathcal{L} S_\mathcal{L} \prod_v \mathcal{L}(v)$$
where: $\mathcal{L}$ ranges over all labelings of $\mathcal{G}_q$; $v$ ranges over all vertices of $\mathcal{G}_q$; the expectation is over the randomness of $\theta^1,\ldots,\theta^K$; and the contributions from $\mathcal{S}$ are captured by
$$S_\mathcal{L} \defeq \prod_{\ell=1}^q S_{a^{(\ell)}b^{(\ell)}c^{(\ell)}d^{(\ell)}}\, S_{(-a^{(\ell+1)})(-b^{(\ell+1)})(-c^{(\ell)})(-d^{(\ell)})}.$$

\begin{definition}
\label{def:repeated-labels}
Define the \emph{number of repeated labels} in a labeling to be
$$c(\mathcal{L}) = \sum_{i \in [p/2]} \max\{0,[\text{\# edges labeled with }\pm i] - 1\}.$$
\end{definition}

\begin{definition}\label{def:region}
For any $k$, call the set of vertices $\{v \;:\; k_v = k\}$ a \emph{region}. The \emph{number of regions} in a labeling is
$$r(\mathcal{L}) = |\{k_v\}_v| = [\text{\# distinct }k_v\text{ values}].$$
\end{definition}

\begin{definition}\label{def:valid}
Call $\mathcal{L}$ a \emph{valid labeling} if $S_\mathcal{L} \,\mathbb{E} \prod_v \mathcal{L}(v) \ne 0$. In other words, $\mathcal{L}$ is valid if and only if
\renewcommand{\theenumi}{\roman{enumi}}
\begin{enumerate}
    \item for every vertex $v$, the three edge labels $i_1,i_2,i_3$ incident to $v$ satisfy $i_1 + i_2 + i_3 = 0$,
    \item for every $\ell$, $a^{(\ell)} \ne -b^{(\ell)}$ and $c^{(\ell)} \ne -d^{(\ell)}$, and
    \item for every $i$, each region has as many incident\footnote{If $R$ is a region and $e = (u,v)$ is an edge with $u \in R$ and $v \notin R$ then the label at the $u$-end of $e$ is considered \emph{incident} to $R$.} $i$ labels as incident $-i$ labels.
\end{enumerate}
\end{definition}
\noindent Recall that we set certain $S_{abcd}$ values to zero, which results in rule (ii) above.

\begin{lemma}
\label{lem:rc}
For any valid labeling $\mathcal{L}$ with $r(\mathcal{L}) > 1$, we have $c(\mathcal{L}) \ge r(\mathcal{L})/2$.
\end{lemma}
\begin{proof}
Since $r(\mathcal{L}) > 1$, every region has at least two edges crossing its boundary\footnote{It is immediate from Definition~\ref{def:valid}(iii) that a region must have an even number of edges crossing its boundary. In fact, it is not hard to see that $\mathcal{G}_q$ has no cuts of size two and so at least four edges must cross.}, so there are at least $r(\mathcal{L})$ such boundary edges total. In a valid labeling, each of these boundary edges must have the same label as at least one other boundary edge. This results in at least $r(\mathcal{L})/2$ repeated labels.
\end{proof}

The following key lemma is proved in Appendix~\ref{app:count}.
\begin{lemma}
\label{lem:count-labelings}
The number of valid edge-labelings\footnote{By \emph{edge-labelings} we mean choices for the edge labels $i_e$ but not the vertex labels $v_k$. Here \emph{valid} means that (i) and (ii) in Definition~\ref{def:valid} are satisfied.} with exactly $c(\mathcal{L})$ repeated labels is at most $3 [2(27q)^2]^{c(\mathcal{L})} p^{1 + 9q - c(\mathcal{L})/25}$.
\end{lemma}
\noindent The interpretation of this is as follows. The important factor is $p^{1 + 9q - c(\mathcal{L})/25}$; the rest is lower-order. Without requiring any repeated labels, the number of valid edge-labelings is $\sim p^{1+9q}$. Thus the lemma shows that when repeated labels are required, the number of valid edge-labelings decreases substantially.

\begin{remark}
To achieve heterogeneity $K \lesssim p^\delta$, one needs to show that the number of valid labelings with $r(\mathcal{L})$ regions is $\lesssim p^{1+9q-\delta r(\mathcal{L})}$. Together, Lemmas~\ref{lem:rc} and \ref{lem:count-labelings} show this for some constant $\delta > 0$, but we have not attempted to optimize $\delta$. See Appendix~\ref{app:design} for more on this.
\end{remark}

Using the above lemmas, we are able to bound $\mathbb{E}[\mathrm{Tr}((\tilde W \tilde W^\top)^q)]$ as desired. We prove the following in Appendix~\ref{app:ETr}.
\begin{lemma}\label{lem:ETr}
With $q = \log p$,
$$\mathbb{E}[\mathrm{Tr}((\tilde W \tilde W^\top)^q)] \le O(1)^q p^2$$
and so
$$\{\mathbb{E}[\mathrm{Tr}((\tilde W \tilde W^\top)^q)]\}^{1/2q} \le O(p^{1/q}) \le O(1).$$
\end{lemma}

By Theorem~\ref{thm:trace-method}, with probability at least $1-1/p$ (over the randomness of $\theta^1,\ldots,\theta^K$) we have $\|\tilde W\| \le O(1)$; let this be the predicate $P_1(\{\theta^k\})$. Provided $P_1(\{\theta^k\})$ holds, we then have by (\ref{eq:rand-contraction}) that
$$\|M(T,\tilde u)\| \le c \sqrt{\log p}$$
with probability at least $1-8p^{2-\Omega(c^2)}$ (over the randomness of $\tilde u$). Taking $c$ to be a sufficiently large constant completes the proof.

\subsection{Heterogeneous signal term}

Here we consider the term $M(T,(\theta^k)^{\otimes 5})-M(T^k,(\theta^k)^{\otimes 5})$. This term is relatively benign compared to the previous one and we bound it using a much simpler variant of the argument in the previous section. For convenience we use $K \le p^\delta$ here but we expect that the previous term (not this one) would be the bottleneck if we wanted to optimize $\delta$.

\begin{proposition} \label{prop:het-signal}
There exists $\delta > 0$ such that if $K \le p^\delta$ then with high probability over $\{\theta^k\}$ we have for every $k \in [K]$,
$$\|M(T,(\theta^k)^{\otimes 5})-M(T^k,(\theta^k)^{\otimes 5})\| \le o(1).$$
\end{proposition}

This section is devoted to proving Proposition~\ref{prop:het-signal}. By Markov's inequality, it is sufficient to show
$$\mathbb{E}\|M(T,(\theta^1)^{\otimes 5})-M(T^1,(\theta^1)^{\otimes 5})\|^2_F \le o(1/K)$$
so that we can take a union bound over all $k \in [K]$.

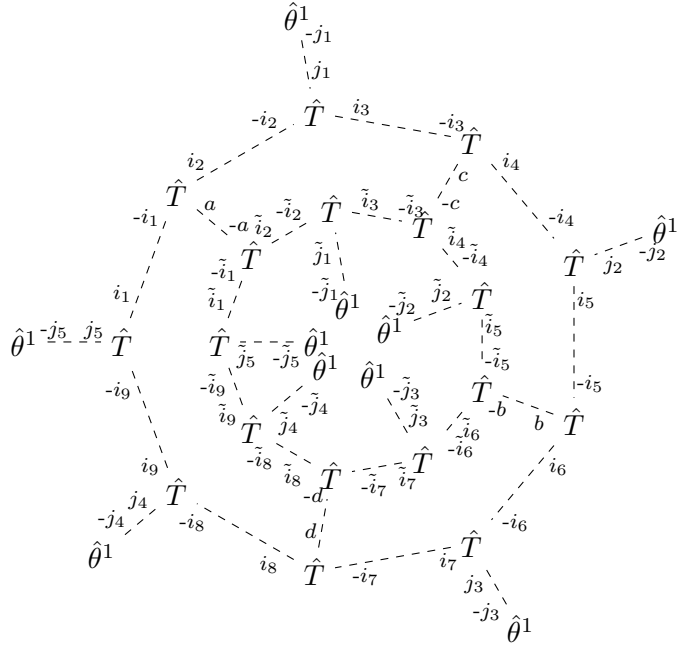
\begin{figure}\centering
\begin{minipage}{0.8\textwidth}\centering
\begin{tikzpicture}[auto]
% radii
\def\ra{0.5}  % theta
\def\rad{1.8}  % T
\def\radi{3.1}  % T
\def\radiu{4.4}  % theta
% inner theta
\foreach \i in {2,4,6,8,9}{
	\node (\i) at ({180-\i*360/9}:\ra) {$\hat{\theta}^1$};
}
% inner T
\foreach \i in {1,...,9}{
	\node (\i+9) at ({180-\i*360/9}:\rad) {$\hat{T}$};
}
% outer T
\foreach \i in {1,...,9}{
	\node (\i+18) at ({180-\i*360/9}:\radi) {$\hat{T}$};
}
% outer theta
\foreach \i in {2,4,6,8,9}{
	\node (\i+27) at ({180-\i*360/9}:\radiu) {$\hat{\theta}^1$};
}
% inner legs ~j
\foreach \i/\j in {2/1,4/2,6/3,8/4,9/5}{
	\draw[dashed] (\i) to node[very near start,inner sep=1pt] {\scriptsize{-$\tilde j_\j$}} node[very near end,inner sep=1pt] {\scriptsize{$\tilde j_\j$}} (\i+9);
}
% outer legs j
\foreach \i/\j in {2/1,4/2,6/3,8/4,9/5}{
	\draw[dashed] (\i+27) to node[very near start,inner sep=1pt] {\scriptsize{-$j_\j$}} node[very near end,inner sep=1pt] {\scriptsize{$j_\j$}} (\i+18);
}
% inner circle ~i
\foreach \i in {1,...,8}{
	\pgfmathtruncatemacro{\nexti}{\i + 1}
	\draw[dashed] (\i+9) to node[very near start,inner sep=1pt] {\scriptsize{$\tilde i_\nexti$}} node[very near end,inner sep=1pt] {\scriptsize{-$\tilde i_\nexti$}} (\nexti+9);
}
\draw[dashed] (9+9) to node[very near start,inner sep=1pt] {\scriptsize{$\tilde i_1$}} node[very near end,inner sep=1pt] {\scriptsize{-$\tilde i_1$}} (1+9);
% outer circle i
\foreach \i in {1,...,8}{
	\pgfmathtruncatemacro{\nexti}{\i + 1}
	\draw[dashed] (\i+18) to node[very near start,inner sep=1pt] {\scriptsize{$i_\nexti$}} node[very near end,inner sep=1pt] {\scriptsize{-$i_\nexti$}} (\nexti+18);
}
\draw[dashed] (9+18) to node[very near start,inner sep=1pt] {\scriptsize{$i_1$}} node[very near end,inner sep=1pt] {\scriptsize{-$i_1$}} (1+18);
% connection abcd
\foreach \i/\a in {1/a,3/c,5/b,7/d}{
	\draw[dashed] (\i+18) to node[very near start,inner sep=1pt] {\scriptsize{$\a$}} node[very near end,inner sep=1pt] {\scriptsize{-$\a$}} (\i+9);
}
\end{tikzpicture}
\caption{The tensor network for $\|M(T,(\theta^1)^{\otimes 5})\|_F^2$. The error term in (\ref{eq:het-frob}) is obtained by only considering the terms corresponding to particular labelings (see Definition~\ref{def:valid-het}).}
\label{fig:het-frob}
\end{minipage}
\end{figure}

The value $\|M(T,(\theta^1)^{\otimes 5})\|_F^2$ is depicted by the tensor network in Figure~\ref{fig:het-frob}. Similarly to the previous section, we consider labelings of Figure~\ref{fig:het-frob}. As before, each edge gets a label $i_e$ and each vertex gets a label $k_v$. (Each $\hat\theta^1$ is also considered a vertex.)

\begin{definition}\label{def:valid-het}
In addition to the requirements in Definition~\ref{def:valid}, we define a \emph{valid labeling} of Figure~\ref{fig:het-frob} to have two additional constraints: (i) each $\hat\theta^1$ has vertex label $k_v = 1$, and (ii) the inner and outer ring of $\hat T$'s each have a vertex for which $k_v \ne 1$.
\end{definition}

By restricting to these valid labelings, we get an expression for the error term that we want. Formally,
\begin{equation}\label{eq:het-frob}
\|M(T,(\theta^1)^{\otimes 5})-M(T^1,(\theta^1)^{\otimes 5})\|^2_F = \sum_{\text{valid }\mathcal{L}} S_{abcd}S_{(-a)(-b)(-c)(-d)} \prod_v \mathcal{L}(v)
\end{equation}
where for a $\hat T$ vertex, $\mathcal{L}(v)$ is defined as in the previous section, and for a $\hat\theta^1$ vertex, $\mathcal{L}(v) = \hat\theta^1_{-j}$ where $-j$ is the incident label.

\begin{lemma}
The number of valid labelings of Figure~\ref{fig:het-frob} is $O(K^{18} p^{13})$.
\end{lemma}
\begin{proof}
There are $K^{18}$ ways to choose the vertex labels for the 18 $\hat T$ vertices. There are at most $p^{14}$ ways to choose the edge labels: $p^4$ ways to choose $a,b,c,d$, $p^8$ ways to choose the $j$'s and $\tilde j$'s such that $-j_1-j_2-\cdots-j_5=a+b+c+d=\tilde j_1+\tilde j_2 + \cdots + \tilde j_5$ (which is required; see Appendix~\ref{app:count}), $p$ ways to choose the $i$'s, and $p$ ways to choose the $\tilde i$'s. We can improve this argument by noting that a valid labeling cannot have all vertex labels equal, so there must be at least two regions (in the sense of Definition~\ref{def:region}) and thus at least one repeated label (in the sense of Definition~\ref{def:repeated-labels}). This constraint removes a degree of freedom in the edge labels (similarly to the argument in Appendix~\ref{app:count}), and so there are only $O(p^{13})$ ways to choose the edge labels.
\end{proof}

Let $S_{\max} = \max_{abcd} S_{abcd} \le O(p^9)$. Since $0 \le \mathbb{E}\prod_v \mathcal{L}(v) \le O(p^{-32})$, we have
\begin{align*}
\mathbb{E}\|M(T,(\theta^1)^{\otimes 5})-M(T^1,(\theta^1)^{\otimes 5})\|^2_F
&\le S_{\max}^2 \sum_{\text{valid }\mathcal{L}}\mathbb{E} \prod_v \mathcal{L}(v)\\
&\le O(p^{18}) \cdot O(K^{18} p^{13}) \cdot O(p^{-32})\\
&\le O(K^{18}/p)
\end{align*}
which is $o(1/K)$ provided $K \le p^{1/20}$.

\subsection{Error term}

Here we consider the term $M(\mathcal{T},u) - M(T,u)$. This is the error term due to the small error $E$ that we allow in our input tensor. We will use very crude and easy bounds here.

\begin{proposition}
\label{prop:error-E}
There is a deterministic predicate $P_2(\{\theta^k\})$ depending only on $\{\theta^k\}$ that occurs with high probability. If $P_2(\{\theta^k\})$ is satisfied then
$$\|M(\mathcal{T},u) - M(T,u)\| \le \tilde O(K^8 p^4 \|E\|_\infty)$$
with overwhelming probability over $u$.
\end{proposition}
\begin{proof}
Let $P_2(\{\theta^k\})$ be the high-probability event that every entry of every $\theta^k$ is bounded by $|\theta^k_i| \le \tilde O(1/\sqrt{p})$. If $P_2(\{\theta^k\})$ is satisfied then every entry of $T$ has size at most $\tilde O(K p^{-3/2})$. Each entry of $u$ has size $\tilde O(1)$ with overwhelming probability. We also have $S_{abcd} \le O(p^9)$. Each entry of $M(\mathcal{T},u)$ is the sum of $O(p^5)$ terms, each of which is the product of: $S_{abcd}$, an entry of $u$, and $9$ entries of $\mathcal{T}$. Each entry of $\mathcal{T}$ has a part from $T$ and a part from $E$; here we only consider the error terms where one of the $9$ entries uses $E$. Each entry of $M(\mathcal{T},u) - M(T,u)$ has size $\tilde O(p^5 \cdot p^9 \cdot (Kp^{-3/2})^8 \cdot \|E\|_\infty) = \tilde O(K^8 p^2 \|E\|_\infty)$, provided $\varepsilon \le Kp^{-3/2}$. Thus
$$\|M(\mathcal{T},u) - M(T,u)\| \le \|M(\mathcal{T},u) - M(T,u)\|_F \le \sqrt{p^4 \,\tilde O(K^8 p^2 \|E\|_\infty)^2} \le \tilde O(K^{8} p^4 \|E\|_\infty).$$
\end{proof}

\subsection{Putting it all together}
\label{sec:together}

Here we complete the proof of Theorem~\ref{thm:mra-tech}. Recall $u \sim \mathcal{N}(0,I) \in \mathbb{R}^{p^5}$, $\Theta^k = (\theta^k \otimes \theta^k)(\theta^k \otimes \theta^k)^\top$ and $u = \alpha (\theta^k)^{\otimes 5} + \tilde u$ with $\tilde u \perp (\theta^k)^{\otimes 5}$. As above, we write
$$M(\mathcal{T},u) = \alpha M(T^k,(\theta^k)^{\otimes 5}) + \alpha[M(T,(\theta^k)^{\otimes 5}) - M(T^k,(\theta^k)^{\otimes 5})] + M(T,\tilde u) + [M(\mathcal{T},u) - M(T,u)].$$

\noindent Let the predicate $P(\{\theta^k\})$ be the intersection of the following high-probability events:
\begin{itemize}
    \item the conclusion of Proposition~\ref{prop:signal} holds for every $k \in [K]$,
    \item $P_1(\{\theta^k\})$ (from Proposition~\ref{prop:noise-u}) holds,
    \item the conclusion of Proposition~\ref{prop:het-signal} holds,
    \item $P_2(\{\theta^k\})$ (from Proposition~\ref{prop:error-E}) holds,
    \item for every $k \in [K]$, $1-\tilde O(1/\sqrt{p}) \le \|\theta^k\| \le 1+\tilde O(1/\sqrt{p})$.
\end{itemize}

\noindent For fixed $\theta^1,\ldots,\theta^K$ satisfying $P(\{\theta^k\})$, we have for any $k$,
\begin{itemize}
\item $\|M(T^k,(\theta^k)^{\otimes 5}) - \Theta^k\| \le o(1),$
\item $\|M(T,(\theta^k)^{\otimes 5})-M(T^k,(\theta^k)^{\otimes 5})\| \le o(1),$
\item $\|M(T,\tilde u)\| \le O(\sqrt{\log p})$ \quad with high probability over $\tilde u$,
\item $\|M(\mathcal{T},u) - M(T,u)\| \le o(1)$ \quad with overwhelming probability over $u$.
\end{itemize}

We have $\alpha = \langle u,(\theta^k)^{\otimes 5} \rangle/\|(\theta^k)^{\otimes 5}\|^2 = \langle u,(\theta^k)^{\otimes 5}/\|(\theta^k)^{\otimes 5}\|\rangle/\|(\theta^k)^{\otimes 5}\| \defeq \tilde \alpha/\|\theta^k\|^5$ where $\tilde \alpha \sim \mathcal{N}(0,1)$ independently from $\tilde u$.

We have the Gaussian lower tail bound
$$\mathrm{Pr}\{\tilde\alpha \ge t\} \ge \frac{1}{2\sqrt{2 \pi}}\, t^{-1} \exp(-t^2/2)$$
and so
\begin{equation}\label{eq:alpha}
\Pr\left\{\tilde\alpha \ge C \sqrt{\log p}\right\} \ge \frac{1}{2\sqrt{2 \pi}}\, \frac{1}{C \sqrt{\log p}} \,p^{-C^2/2}.
\end{equation}
We can write $M_\mathrm{sym}(\mathcal{T},u) \defeq \frac{1}{2}[M(\mathcal{T},u)+M(\mathcal{T},u)^\top] = \alpha \,\Theta^k + B = \tilde\alpha\,\Theta^k/\|\theta^k\|^5 + B$ where
\begin{equation}\label{eq:B}
\|B\| \cdot \|\theta^k\| \defeq \beta \le o(1) \,\alpha + O(\sqrt{\log p}) \le o(1) \,\tilde\alpha + O(\sqrt{\log p}).
\end{equation}

Let $w = (\theta^k \otimes \theta^k)/\|\theta^k\|^2 \in (\mathbb{R}^p)^{\otimes 2}$ so that $ww^\top = \Theta^k/\|\theta^k\|^4$. Let $v \in (\mathbb{R}^p)^{\otimes 2}$ be the leading (unit-norm) eigenvector of $M_\mathrm{sym}(\mathcal{T},u)$, which is also the leading eigenvector of $\tilde M_\mathrm{sym}(\mathcal{T},u) = \|\theta^k\| M_\mathrm{sym}(\mathcal{T},u) = \tilde \alpha ww^\top + \|\theta^k\|B$. We have
$$\tilde\alpha \langle v,w \rangle^2 + \beta \ge v^\top \tilde M_\mathrm{sym}(\mathcal{T},u) \,v \ge w^\top \tilde M_\mathrm{sym}(\mathcal{T},u) \,w \ge \tilde\alpha - \beta$$
and so $$\langle v,w \rangle^2 \ge 1 - \frac{2\beta}{\tilde\alpha}.$$
Re-shape $v$ into a $p \times p$ matrix $\tilde V$ and let $V = \frac{1}{2}(\tilde V + \tilde V^\top)$. Let $\tau$ be the eigenvector of $V$ corresponding to the eigenvalue of largest absolute value. Let $y = \theta^k/\|\theta^k\|$. Write
$$V = \langle V,yy^\top \rangle yy^\top + B'$$
where $\|B'\|^2 \le \|B'\|_F^2 = \|V\|_F - \langle V,yy^\top \rangle^2 \le 1 - \langle V,yy^\top \rangle^2$.
We have
$$|\langle V,yy^\top \rangle| = |y^\top V y| \le |\tau^\top V \tau| \le |\langle V,yy^\top \rangle| \cdot\langle \tau,y \rangle^2 + \|B'\|$$
and so
$$\langle \tau,y \rangle^2 \ge 1 - \frac{\|B'\|}{|\langle V,yy^\top \rangle|} \ge 1 - \frac{\sqrt{1-\langle V,yy^\top \rangle^2}}{|\langle V,yy^\top \rangle|}.$$
Note that
$$\langle V, yy^\top \rangle^2 = \langle \tilde V, yy^\top \rangle^2 = \langle v,w \rangle^2 \ge 1 - \frac{2\beta}{\tilde\alpha}$$
and so, provided $2\beta/\tilde\alpha \le 1/2$,
$$\langle \tau,y \rangle^2 \ge 1 - \frac{\sqrt{2\beta/\tilde\alpha}}{\sqrt{1/2}} = 1 - 2 \sqrt{\frac{\beta}{\tilde\alpha}}.$$
Recall from (\ref{eq:B}) that if $\tilde\alpha \ge C\sqrt{\log p}$ then $\beta/\tilde\alpha \le o(1) + O(1)/C$. Thus, to have $\langle \tau,y \rangle^2 \ge 1 - \varepsilon - o(1)$, it is sufficient to take $C = O(1)/\varepsilon^2$. Using (\ref{eq:alpha}), the success probability is $\ge p^{-O(1)/\varepsilon^4}$.

\appendix

\section{Omitted proofs}

\subsection{Proof of Lemma~\ref{lem:Es}}
\label{sec:pf-Es}

Here we prove the lower bound in Lemma~\ref{lem:Es}: for some constant $c_1$, $\mathbb{E}[s_{abcd}] \ge c_1 p^{-9}$.

\begin{proof}[Proof of Lemma~\ref{lem:Es}]
For every $a,b,c,d$, we will show (by explicit construction) that there are at least $c_1 p^5$ (nonzero) terms in $s_{abcd}$. This completes the proof because each term has expectation at least $p^{-14}$.

Once $i_1,j_1,j_2,j_3,j_4$ are fixed, we have $i_2 = i_1 - a, i_3 = i_2 - j_1, i_4 = i_3 - c, \ldots, i_9 = i_8 - j_4$. We then need to pick $j_5$ so that $i_1 = i_9 - j_5$.

Let us first construct one solution. Regardless of $i_1$, it is possible by induction to choose $j_1,j_2,j_3 \in [-p/2,p/2]$ such that $i_2,\ldots,i_8$ lie between $i_1$ and $i_1 - \frac{p}{2}\,\sgn(a)$ (where $\sgn(a) = 1$ if $a \ge 0$ and $\sgn(a) = -1$ if $a < 0$). We can then take $j_4 \in [-p/2,p/2]$ so that $i_9 = i_1$, and $j_5 = 0$. Since $i_1,\ldots,i_9$ all lie in an interval of length $p/2$, we now choose $i_1$ anywhere in a particular interval of length $p/2$ so that $i_1,\ldots,i_9$ all lie in $[-p/2,p/2]$.

Now consider perturbations of the above solution where we allow $j_1,\ldots,j_4$ to deviate from their above values by at most $p/10$. This causes $i_2,\ldots,i_9$ to each deviate from their original value by at most $2p/5$. We pick $j_5 \in [-2p/5,2p/5]$ in order to compensate. Now $i_1,\ldots,i_9$ lie in an interval of length $9p/10$, so the interval of possible $i_1$ values now has length $\ge p/10$. This yields $(10^{-5} - o(1))p^5$ different solutions.
\end{proof}

\subsection{Proof of Lemma~\ref{lem:count-labelings}}
\label{app:count}

\begin{proof}[Proof of Lemma~\ref{lem:count-labelings}]
By summing the vertex constraints (from Definition~\ref{def:valid}(i)) in a ring, we have
$$a^{(\ell)}+b^{(\ell)} = -c^{(\ell)}-d^{(\ell)}-j_1^{(\ell)} - j_2^{(\ell)} - j_3^{(\ell)} - j_4^{(\ell)} - j_5^{(\ell)} = a^{(\ell+1)}+b^{(\ell+1)} \defeq w$$
for all $\ell$. There are $2p+1$ choices for $w$.

Within each layer, split the edges into \emph{classes} $\{ab\}$ (2 edges) and $\{cdi\tilde ij\}$ (25 edges). In total there are $q$ copies of each class. After fixing $w$, imagine picking the labels sequentially, where first we pick all the $\{ab\}$ labels layer by layer, and then we pick all the $\{cdi\tilde ij\}$ labels layer by layer. (Within each class, fix an arbitrary order for the labels.) Call a label \emph{repeated} if it is equal (up to sign) to a previous label in the above ordering. The total number of repeated labels is, by definition, $c(\mathcal{L})$.

Let $N = 27q$ be the total number of edges. There are at most $N^{c(\mathcal{L})}$ ways to pick which edges are \emph{repeated}, and then at most $(2N)^{c(\mathcal{L})}$ ways to pick which previous edges they repeat (where the factor of 2 is for the choice of sign). We imagine choosing this structure of repeats in advance and then picking all the labels subject to these constraints.

Without any repeats imposed, each $\{ab\}$ class has 1 \emph{free label}, i.e.\ $a$ is free to take any value but then we must set $b = w-a$. Each $\{cdi \tilde ij\}$ class has 8 free labels: there are 6 free labels among $c,d,j_1,\ldots,j_5$ (since they must sum to $-w$), 1 free label among the $i$'s (since once one $i$ label is chosen, this propagates around the ring), and similarly 1 free label among the $\tilde i$'s. Overall, there are 9 free labels per layer.

Each class ($\{ab\}$ or $\{cdi\tilde ij\}$) has $\le 25$ edges. At least $c(\mathcal{L})/25$ classes must have a repeated label. Each class with a repeated label has (at least) one fewer free label than it would otherwise have; see Lemma~\ref{lem:free} below. So the total number of free labels is $q + 8q - c(\mathcal{L})/25 = 9q - c(\mathcal{L})/25$.

Thus the total number of valid edge-labelings is at most
$$(2p+1) (2N^2)^{c(\mathcal{L})} p^{9q - c(\mathcal{L})/25} \le 3(2N^2)^{c(\mathcal{L})} p^{1 + 9q - c(\mathcal{L})/25}.$$
\end{proof}

\begin{lemma}\label{lem:free}
Suppose we have already chosen $w$ and the structure of repeated labels and now want to count the number of labelings (subject to these constraints). Each $\{ab\}$ class with a repeated label has no free labels, i.e.\ at most 1 possible labeling (for any fixed values of the previous labels). Each $\{cdi\tilde ij\}$ class with a repeated label has at most 7 free labels, i.e.\ at most $p^7$ possible labelings.
\end{lemma}
\begin{proof}
First consider an $\{ab\}$ class. Here we have the constraint $a+b=w$. If either $a$ or $b$ is constrained to repeat a label from a previous class, then there is only 1 possible way to label $a$ and $b$. If the repeat $a=b$ is constrained then the only possibility is $a = b = w/2$. If $a = -b$ is constrained then there are no possible labelings because this would violate Definition~\ref{def:valid}(ii). (This is why we needed to set $S_{abcd}$ to zero when $a=-b$ or $c=-d$. If it were not for this, when $w=0$ we would have the repeat $a=-b$ on every layer without sacrificing any free labels.)

Now consider a $\{cdi\tilde ij\}$ class. There are a few cases to check. It is clear that a free label is lost in each of the following cases.
\begin{itemize}
    \item A label in the class is constrained to repeat a label from a previous class.
    \item Two labels from $c,d,j$ are constrained to be equal (up to sign).
    \item A label from $c,d,j$ is constrained to equal (up to sign) an $i$ or $\tilde i$ label. (This removes the free label from $i$ or $\tilde i$.)
    \item An $i$ label is constrained to repeat a $\tilde i$ label. (This removes the free label from $\tilde i$.)
\end{itemize}
There is one more subtler case to check: a repeat is constrained within a single ring of $i$'s (or $\tilde i$'s). If $i_{\ell_1} = -i_{\ell_2}$ is constrained then (once $a,b,c,d,j$ are chosen) there is only 1 possible value for $i_{\ell_1}$ (i.e.\ the ring has no free labels). Suppose instead that $i_{\ell_1} = i_{\ell_2}$ is constrained. If $i_{\ell_1}$ and $i_{\ell_2}$ are adjacent in the ring then there is no possible labeling since the edge between them would need a zero label (which is not in $\pm [p/2]$). Otherwise $i_{\ell_1}$ and $i_{\ell_2}$ split the $c,d,j$ labels into two nonempty sets, each of which needs to sum to a particular value; thus a free label is lost from $c,d,j$.
\end{proof}

\subsection{Proof of Lemma~\ref{lem:ETr}}\label{app:ETr}

\begin{proof}[Proof of Lemma~\ref{lem:ETr}]
For a valid labeling $\mathcal{L}$, we have
$$0 \le \mathbb{E}\prod_v \mathcal{L}(v) \le p^{-27q} (c(\mathcal{L})+1)!$$
using Lemma~\ref{lem:moments}. By Lemma~\ref{lem:rc}, a valid labeling has at most $2(c(\mathcal{L})+1)$ regions (where the $+1$ accommodates the case $c(\mathcal{L})=0$). There are at most $[2(c(\mathcal{L})+1)]^{18q}$ ways to decide which of the $18q$ vertices belong to each region, and $K^{2(c(\mathcal{L})+1)}$ ways to assign $k_v$-values to the regions.

Let $S_{\max} = \max_{abcd} S_{abcd} \le O(p^9)$. Since $S_{abcd} \ge 0$ and $\mathbb{E}\prod_v \mathcal{L}(v) \ge 0$, we have
\begin{align*}
\mathbb{E}[\mathrm{Tr}((\tilde W \tilde W^\top)^q)] &= \mathbb{E}\sum_\mathcal{L} S_\mathcal{L} \prod_v \mathcal{L}(v)\\
&\le S_{\max}^{2q} \sum_{\text{valid }\mathcal{L}} \mathbb{E}\prod_v \mathcal{L}(v)\\
&\le O(p^9)^{2q} \sum_{c=0}^{27q} 3 [2(27q)^2]^c p^{1 + 9q - c/25} \cdot p^{-27q} (c+1)! \cdot (2(c+1))^{18q} K^{2(c+1)}\\
&\le O(1)^q \sum_{c=1}^{27q+1} q^{2c} p^{1 - (c-1)/25} c^c c^{18q} K^{2c}\\
&\le O(1)^q p^2 \sum_{c=0}^\infty (q^3 p^{-1/25} K^2)^c c^{18q}.\\
\end{align*}

Set $q = \log p$. We have a series of the form $\sum_{c=0}^\infty R^c c^N$ with $R = q^3 p^{-1/25} K^2$ and $N = 18q$. Fix $\eta > 0$ and require $K \le p^{1/50-\eta}$ so that $R \le \tilde O(p^{-2\eta})$. Now
$$\sum_{c=0}^\infty R^c c^N = \sum_{c=0}^{20/\eta} R^c c^N + \sum_{c=20/\eta}^\infty R^c c^N \le O(1)^q.$$
Here the second sum was evaluated by noting that the first term is $O(1)^q$ and the ratio between successive terms is
$$R(1+1/c)^N \le \tilde O(p^{-2\eta})(1+\eta/20)^{18\log p} \le \tilde{O}(p^{-2\eta+18\log(1+\eta/20)}) \le \tilde{O}(p^{-\eta}) \le 1/2.$$
\end{proof}

\section{Aside: tensor PCA}\label{app:tensor-pca}

We remark that many popular algorithms for the related problem of tensor PCA (principal component analysis) can also be interpreted as spectral methods on tensor networks. (Third-order) tensor PCA \cite{RM-tensor-pca} is the problem of recovering a unit vector $x \in \mathbb{R}^p$ given the tensor $T = \lambda x^{\otimes 3} + W$ where $\lambda \ge 0$ is a signal-to-noise parameter and $W$ has entries drawn i.i.d.\ from $\mathcal{N}(0,1)$. The problem is statistically possible for $\lambda = O(\sqrt{p})$ but it is expected that polynomial-time algorithms require $\lambda \gtrsim p^{3/4}$. Various methods are known to achieve $\lambda \sim p^{3/4}$; we now briefly describe a few of these and interpret them as spectral methods on tensor networks (see Figure~\ref{fig:tensor-pca}).

\begin{itemize}
    \item {\bf Tensor unfolding}: This method flattens $T$ to a $p \times p^2$ matrix $\tilde T$ and then computes the top eigenvector of $\tilde T \tilde T^\top$. In other words, we compute the top eigenvector of the $p \times p$ matrix described by the tensor network in Figure~\ref{fig:tensor-pca}(a). This method was shown to work when $\lambda \gtrsim p$ by \cite{RM-tensor-pca} and later for $\lambda \gtrsim p^{3/4}$ by \cite{tensor-pca-sos}.
    \item {\bf Spectral SoS}: Inspired by their algorithm based on the sum-of-squares (SoS) hierarchy, \cite{tensor-pca-sos} give the following spectral method. Let $T_i$ denote the $i$th slice of $T$, i.e.\ the $p \times p$ matrix $(T_i)_{jk} = T_{ijk}$. Compute the top eigenvector of the $p^2 \times p^2$ matrix $\sum_i (T_i \otimes T_i)$. This matrix is the $(\{a,b\},\{c,d\})$-flattening of the tensor network in Figure~\ref{fig:tensor-pca}(b). (The above method gives an estimate for $x^{\otimes 2}$, from which $x$ can be extracted via another eigenvector computation.)
    \item {\bf Spectral SoS with partial trace}: In \cite{HSSS-fast-sos}, an improvement to the above method is given. By applying the partial trace operation, they reduce the matrix from $p^2 \times p^2$ to $p \times p$, giving a speedup in runtime. Specifically, they compute the top eigenvector of the matrix $\sum_i \mathrm{Tr}(T_i) T_i$, which is given by the tensor network in Figure~\ref{fig:tensor-pca}(c).
    \item {\bf Homotopy method}: This is a type of local search algorithm analyzed by \cite{homotopy}. A key step in the analysis is to obtain a good initialization. Specifically, they initialize the method with the vector $z_j = \sum_i T_{iij}$, whose tensor network is shown in Figure~\ref{fig:tensor-pca}(d).
\end{itemize}

% figure: tensor PCA
\begin{figure}\centering
\begin{minipage}{0.8\textwidth}\centering
\begin{tikzpicture}
\tikzset{every loop/.style={}}
\node at (0.25,3.5) {\bf (a)};
\coordinate (top) at (1,3.3) {};
\node (T1) at (1,2.3) {$T$};
\node (T2) at (1,1) {$T$};
\coordinate (bot) at (1,0) {};
\draw (T1) -- (top);
\draw[transform canvas={xshift=-1.5pt}] (T1) -- (T2);
\draw[transform canvas={xshift=1.5pt}] (T1) -- (T2);
\draw (T2) -- (bot);
\begin{scope}[shift={(3,0)}]
\node at (-0.5,3.5) {\bf (b)};
\node (T1) at (1,2.3) {$T$};
\node (T2) at (1,1) {$T$};
\coordinate (a) at ({1-sin(60)},{2.3+cos(60)}) {};
\coordinate (c) at ({1+sin(60)},{2.3+cos(60)}) {};
\coordinate (b) at ({1-sin(60)},{1-cos(60)}) {};
\coordinate (d) at ({1+sin(60)},{1-cos(60)}) {};
\draw (T1) -- (T2);
\draw (T1) -- (a) node [midway, above] {$a$};
\draw (T1) -- (c) node [midway, above] {$c$};
\draw (T2) -- (b) node [midway, below] {$b$};
\draw (T2) -- (d) node [midway, below] {$d$};
\end{scope}
\begin{scope}[shift={(6.5,0)}]
\node at (-0.5,3.5) {\bf (c)};
\node (T1) at (1,2.3) {$T$};
\node (T2) at (1,1) {$T$};
\coordinate (a) at ({1-sin(60)},{2.3+cos(60)}) {};
\coordinate (c) at ({1+sin(60)},{2.3+cos(60)}) {};
\draw (T1) -- (T2);
\draw (T1) -- (a);
\draw (T1) -- (c);
\draw (T2) edge[loop below] (T2);
\end{scope}
\begin{scope}[shift={(9.25,0)}]
\node at (0.25,3.5) {\bf (d)};
\node (T) at (1,1.5) {$T$};
\coordinate (top) at (1,2.5);
\draw (T) -- (top);
\draw (T) edge[loop below] (T);
\end{scope}
\end{tikzpicture}
\caption{Tensor networks used in prior work for tensor PCA. {\bf (a)} Tensor unfolding. {\bf (b)} Spectral SoS: $(\{a,b\},\{c,d\})$-flattening. {\bf (c)} Spectral SoS with partial trace. {\bf (d)} Initialization for homotopy method.}
\label{fig:tensor-pca}
\end{minipage}
\end{figure}
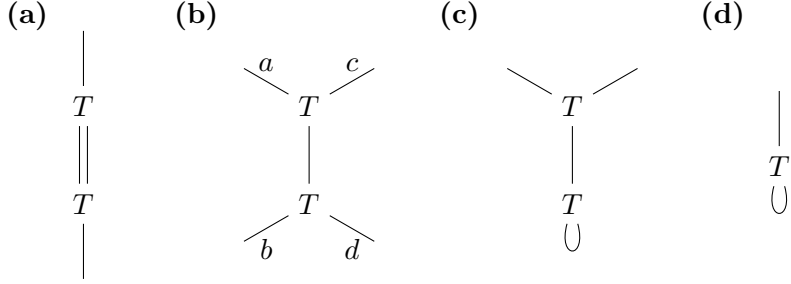

\section{Design of tensor network}
\label{app:design}

In this section we explain some of the considerations that went into the particular design of the tensor network depicted in Figure~\ref{fig:net-9}. In particular, we explain why various simpler tensor networks would not work for our purposes.

Let us first explain why the tensor network in Figure~\ref{fig:net-intro}(c) (which was used for random overcomplete tensor decomposition \cite{HSSS-fast-sos}) fails for continuous MRA. Recall that our input tensor $\hat T_{ijk}$ is supported on entries where $i+j+k=0$. We cannot have two output legs (e.g.\ $a,c$) connected to a single copy of $T$ because e.g.\ if we take $a=c=p/2$ then for all values $i \in \pm [p/2]$ for the third index, $T_{aci} = 0$. For this reason, the tensor network in Figure~\ref{fig:net-intro}(c) produces a matrix for which a constant fraction of entries are zero.

While \cite{HSSS-fast-sos} takes the random guess $u$ to be a vector, we now justify why in our case $u$ needs to be a higher order tensor. By summing the zero-sum constraints over all vertices in Figure~\ref{fig:net-mra}, we obtain $a+b+c+d+j_1+\cdots+j_5=0$. Thus, for the $abcd$ entry to be nonzero, there must exist $j_1,\ldots,j_5 \in \pm [p/2]$ with sum $-(a+b+c+d)$. Note that for $a=b=c=d=p/2$, this would not be possible if there were fewer than four $j$'s (i.e.\ if $u$ had order less than four). For technical convenience, we choose to take $u$ of order 5 rather than 4 so that each entry $abcd$ has many possible settings for the $j$'s.

The above considerations turn out to be sufficient for handling the signal term (treated in Section~\ref{sec:signal}). So long as we take the above precautions to avoid having many zero entries in our matrix, we can show that if we were to guess the correct $u$ then (after a suitable correction) our matrix is close to the rank-1 signal we desire.

There are a few subtler considerations involved in spectrally bounding the noise term (treated in Section~\ref{sec:noise}). The noise term comes from the ``bad'' component of $u$ that is orthogonal to the correct guess $(\theta^k)^{\otimes 5}$. As explained in Section~\ref{sec:noise}, by using the trace moment method, the analysis boils down to a combinatorial question about labelings of an expanded tensor network $\mathcal{G}_q$ made up of $2q = 2 \log p$ copies of the original network (shown in Figure~\ref{fig:mra-trace}). Roughly speaking, the combinatorial question is as follows. Consider edge-labelings of $\mathcal{G}_q$ for which the three labels incident to each vertex must sum to zero. Suppose the number of such labelings is (roughly) $p^{f(0)}$, i.e.\ $f(0)$ is the number of \emph{free labels} that can take any value, and then the remaining labels are uniquely determined by the zero-sum constraints. Now consider a partition of the vertices of $\mathcal{G}_q$ and call each part of the partition a \emph{region} (regions need not be connected in the graph). We require regions to obey the following constraint: the multiset of labels incident to all vertices in a region must have an equal number of copies of $i$ and $-i$ (for any $i$). Since edges internal to the region preserve this balance (one end is $i$ and the other is $-i$), we only need to ensure that the balance is correct for edges on the boundary. Imagine we fix the regions and then count the number of valid labelings; due to the constraints imposed by regions, the number of free labels may now be less than $f(0)$. Let $f(r)$ be the maximum number of free labels over partitions with exactly $r$ (nonempty) regions. With the above setup in mind, the key combinatorial question is captured by the following informal claim.
\begin{claim}[informal]
For heterogeneous continuous MRA, if a tensor network satisfies $f(r) \le f(0) - \delta r$ for all $r$ then the noise term is spectrally bounded provided $K \lesssim p^\delta$.
\end{claim}
In this paper we only show that this holds for some $\delta > 0$ and do not attempt to optimize the constant $\delta$. In particular, we need to ensure that it is not possible to ``pay'' only $O(1)$ free labels to create $\Omega(q)$ regions. This requirement dictates a few additional properties of our tensor network, as we now describe.

The reason we need to set some $S_{abcd}$ to zero is because otherwise we can pay a single free label in order to set $a^{(1)} = -b^{(1)}$, and this propagates automatically  to $a^{(\ell)} = -b^{(\ell)}$ at every layer, allowing each layer to be its own region.

It is essential that our tensor network has a cycle because otherwise it would need to have a $T$ that is adjacent to only one other $T$. This causes each layer of $\mathcal{G}_q$ to have two vertices with a double-edge between them. These two vertices can be their own region without losing any free labels.

One can check that the particular tensor network used in this paper cannot achieve $\delta$ larger than $2/9$. To see this, create 9 regions (each with two vertices) in each layer by pairing each vertex in the $i$ ring with the corresponding vertex in the $\tilde i$ ring. This configuration costs 2 free labels per layer (e.g.\ set $a^{(1)} = a^{(2)}$ and $i_1 = -\tilde i_1$).

We can also see that no tensor network will be able to achieve $\delta$ larger than $1/2$ (which matches the conjectured optimal heterogeneity achievable by efficient algorithms\footnote{We conjecture this based on both analogy to discrete MRA \cite{mra-het} and analogy to tensor completion \cite{alex-thesis}.}). Imagine splitting $\mathcal{G}_q$ into two connected subgraphs (each consisting of half the layers) and pairing each vertex in the first subgraph with the corresponding vertex in the second subgraph (to create regions of size 2 that are not connected). The number of regions is the number $N$ of vertices in one subgraph, and the number of free labels lost is the number of free labels in one subgraph, which is $\sim N/2$ (since there are $\sim 3N/2$ edges and $\sim N$ zero-sum constraints).

\section*{Acknowledgements}

This work was inspired by ideas of Amelia Perry, who hoped to use tensor networks to find improved sum-of-squares algorithms for the planted clique problem (before this was shown to be impossible \cite{sos-clique}). Tragically, Amelia passed away in January, 2018.

\bibliographystyle{alpha}
\bibliography{main}

\end{document}